\documentclass[11pt, letterpaper]{amsart}
\usepackage{graphicx, amssymb, color}
\usepackage{amsmath}
\usepackage{graphicx}
\usepackage{enumerate}

\newtheorem{thm}{Theorem}[section]
\newtheorem{cor}[thm]{Corollary}
\newtheorem{lem}[thm]{Lemma}
\newtheorem{prop}[thm]{Proposition}
\theoremstyle{definition}
\newtheorem{defn}[thm]{Definition}

\newtheorem{ass}[thm]{Assumption}

\theoremstyle{remark}
\newtheorem{rem}[thm]{Remark}

\numberwithin{equation}{section}

\newcommand{\Real}{\mathbb R}
\newcommand{\Natural}{\mathbb N}
\newcommand{\Integer}{\mathbb Z}
\newcommand{\cadlag}{c\`adl\`ag}
\newcommand{\F}{\mathcal{F}}

\newcommand{\cL}{\mathcal{L}}

\newcommand{\prob}{\mathbb{P}}
\newcommand{\qprob}{\mathbb{Q}}
\newcommand{\expec}{\mathbb{E}}

\newcommand{\indic}{\mathbb{I}}
\newcommand{\pare}[1]{\left(#1\right)}
\newcommand{\bra}[1]{\left[#1\right]}

\newcommand{\such}{\, | \, }

\newcommand{\tv}{\tilde{v}}
\newcommand{\wc}{\stackrel{\mathcal{L}}{\longrightarrow}}
\newcommand{\dm}{\underline{m}_n}
\newcommand{\um}{\overline{m}_n}
\newcommand{\dmd}{\underline{m}_n^\delta}
\newcommand{\umd}{\overline{m}_n^\delta}
\newcommand{\dId}{\underline{I}^{\delta,n}}
\newcommand{\uId}{\overline{I}^{\delta,n}}
\newcommand{\uLd}{\cL^{\delta, \umd}}
\newcommand{\dLd}{\cL^{\delta, \dmd}}

\title{Asymptotic Glosten Milgrom equilibrium}\thanks{The authors are grateful to Luciano Campi, Umut {\c{C}}etin, as well as the anonymous Associate Editor and two referees for their valuable comments which help us improving this paper.}

\author[]{Cheng Li}

\author[]{Hao Xing}
\address[Cheng Li and Hao Xing]{Department of Statistics,
London School of Economics and Political Science,
10 Houghton st,
London, WC2A 2AE,
UK}
\email{c.li25@lse.ac.uk, h.xing@lse.ac.uk}

\begin{document}

\begin{abstract}
 This paper studies the Glosten Milgrom model whose risky asset value admits an arbitrary discrete distribution. Contrast to existing results on insider's models, the insider's optimal strategy in this model, if exists, is not of feedback type. Therefore a weak formulation of equilibrium is proposed. In this weak formulation, the inconspicuous trade theorem still holds, but the optimality for the insider's strategy is not enforced. However, the insider can employ some feedback strategy whose associated expected profit is close to the optimal value, when the order size is small. Moreover this discrepancy converges to zero when the order size diminishes. The existence of such a weak equilibrium is established, in which the insider's strategy converges to the Kyle optimal strategy when the order size goes to zero.
\end{abstract}


\maketitle

\section{Introduction} \label{sec: intro}

In the theory of market microstructure, two models, due to Kyle \cite{Kyle} and Glosten and Milgrom \cite{Glosten-Milgrom}, are particularly influential. In the Kyle model, buy and sell orders are batched together by a market maker, who sets a unique price at each auction date. In the Glosten Milgrom model, buy and sell orders are executed by the market maker individually, hence bid and ask prices appear naturally. In both models\footnote{A profit maximizing informed agent is introduced in the Glosten Milgrom model in \cite{Back-Baruch}.}, an informed agent (insider) trades to maximize her expected profit utilizing her private information on the asset fundamental value, while another group of noise traders trade independently of the fundamental value. The cumulative demand of these noise traders is modeled by a Brownian motion in Kyle model, cf. \cite{Back}, and by the difference of two independent Poisson processes, whose jump size is scaled by the order size, in the Glosten Milgrom model.

When the fundamental value, described by a random variable $\tv$, has an arbitrary continuous distribution\footnote{Models with discrete distributed $\tv$ can be studied similarly as in \cite{Back}.}, Back \cite{Back} establishes a unique equilibrium between the insider and the market maker. Moreover, the cumulative demand process in the equilibrium connects elegantly to the theory of filtration enlargement, cf. \cite{Mansuy-Yor}. However much less is known about equilibrium in the Glosten Milgrom model. Back and Baruch \cite{Back-Baruch} consider a Bernoulli distributed $\tv$. In this case, the insider's optimal strategy is constructed in \cite{Cetin-Xing}. Equilibrium with general distribution of $\tv$, as Cho \cite{Cho} puts it, ``will be a great challenge to consider''.

In this paper, we consider the Glosten Milgrom model whose risky asset value $\tv$ has a discrete distribution:
\begin{equation}\label{eq: v dist}
 \prob(\tilde{v}=v_n) =p_n, \quad n=1, \cdots, N,
\end{equation}
where $N\in \Natural \cup \{\infty\}$, $(v_n)_{n=1, \cdots, N}$ is an increasing sequence and $p_n\in (0,1)$ with $\sum_{n=1}^N p_n=1$.  This generalizes the setting in \cite{Back-Baruch} where $N=2$ is considered, i.e., $\tv$ has a Bernoulli distribution.

In models of insider trading, \emph{inconspicuous trade theorem} is commonly observed, cf. e.g., \cite{Kyle}, \cite{Back}, \cite{BP}, \cite{Back-Baruch}, \cite{Campi-Cetin}, and \cite{Campi-Cetin-Danilova} for equilibria of Kyle type, and \cite{Cetin-Xing} for the Glosten Milgrom equilibrium with Bernoulli distributed fundamental value. The inconspicuous trade theorem states, when the insider is trading optimally in equilibrium, the cumulative net orders from both insider and noise traders have the same distribution as the net orders from noise traders, i.e., the insider is able to hide her trades among noise trades. As a consequence, this allows the market maker to set the trading price only considering current cumulative noise trades. Moreover, in all aforementioned studies, the insider's optimal strategy is of \emph{feedback form}, which only depends on the current cumulative total order. This functional form is associated to optimizers of the Hamilton-Jacobi-Bellman (HJB) equation for the insider's optimization problem. However the situation is dramatically different in the Glosten Milgrom model with $N$ in \eqref{eq: v dist} at least $3$. Theorem \ref{thm: nonexistence} below shows that, given aforementioned pricing mechanism, the insider's optimal strategy, if exists, does \emph{not} correspond to optimizers of the HJB equation. This result roots in the difference between bid and ask prices in the Glosten Milgrom model, which is contrast to the unique price in the Kyle model.

Therefore to establish equilibrium in these Glosten Milgrom models, we propose a weak formulation of equilibrium in Definition \ref{def: Asy GM eq}, which is motivated by the convergence of Glosten Milgrom equilibria to the Kyle equilibrium, as the order size diminishing and the trading intensities increasing to infinity, cf. \cite{Back-Baruch} and \cite{Cetin-Xing}. In this weak formulation, the insider still trades to enforce the inconspicuous trading theorem, but the insider's strategy may not be optimal. However, the insider can employ some feedback strategy so that the loss to her expected profit (compared to the optimal value) is small for small order size. Moreover this gap converges to zero when the order size vanishes. We call this weak formulation \emph{asymptotic Glosten Milgrom equilibrium} and establish its existence in Theorem \ref{thm: main}.

In the asymptotic Glosten Milgrom equilibrium, the insider's strategy is constructed explicitly in Section \ref{sec: point process bridge}, using a similar construction as in \cite{Cetin-Xing}. Using this strategy, the insider trades towards a middle level of an interval, driving the total demand process into this interval at the terminal date. This bridge behavior is widely observed in the aforementioned studies on insider trading. On the other hand, the insider's strategy is of feedback form. Hence the insider can determine her trading intensity only using the current cumulative total demand. Moreover, as order size diminishes, the family of suboptimal strategies converge to the optimal strategy in Kyle model, cf. Theorem \ref{thm: str conv}. In such an asymptotic Glosten Milgrom equilibrium, the insider loses some expected profit. The expression of this profit loss is quite interesting mathematically: it is the difference of two stochastic integrals with respect to (scaled) Poisson occupation time. As the order size vanishes, both integrals converge to the same stochastic integral with respect to Brownian local time, hence their difference vanishes.

The paper is organized as follows. Main results are presented in Section \ref{sec: main results}. The mismatch between insider's optimal strategy and optimizers for the HJB equation is proved in Section \ref{sec: nonexistence}. Then a family of suboptimal strategies are characterized and constructed in Sections \ref{sec: suboptimal} and \ref{sec: point process bridge}. Finally the existence of asymptotic equilibrium is established in Section \ref{sec: convergence} and a technical result is proved in Appendix.

\section{Main results} \label{sec: main results}
\subsection{The model}

We consider a continuous time market for a risky and a risk free asset. The risk free interest rate is normalized to $0$, i.e., the risk free asset is regarded as the num\'{e}raire. We assume that the \emph{fundamental value} of the risky asset $\tv$ has a discrete distribution of type \eqref{eq: v dist}. This fundamental value will be revealed to all market participants at a finite time horizon, say $1$, at which point the market will terminate.

The micro-structure of the market and the interaction of market participants are modeled similarly to \cite{Back-Baruch} which we recall below. There are three types of agents: uninformed/noise traders, an informed trader/insider, and a market maker, all of whom are risk neutral. These agents share the same view toward future randomness of the market, but they possess different information. Therefore, the probability space $(\Omega, \prob)$ with different filtration  accommodates the following processes:

\begin{itemize}
\item \emph{Noise traders} trade for liquidity or hedging reasons which are independent of the fundamental value $\tv$. The cumulative demand $Z$ is described by the difference of two independent jump processes $Z^B$ and $Z^S$ which are the cumulative buy and sell orders, respectively. Therefore $Z=Z^B-Z^S$ and it is independent of $\tv$. Noise traders only submit orders of fixed sized $\delta$ every time they trade. As in \cite{Back-Baruch}, $Z^B/\delta$ and $Z^S/\delta$ are assumed to be independent Poisson processes with constant intensity $\beta$. Let $(\F_t^Z)_{t \in [0,1]}$ be the smallest filtration generated by $Z$ and satisfying the usual conditions. Then $(\F^Z_t)_{t \in [0,1]}$ describes the information structure of noise traders.

\item The \emph{insider} knows the fundamental value $\tilde{v}$ at time 0 and observes the market price for the risky asset between time 0 and 1. The insider also submits orders of fixed size $\delta$ in every trade and tries to maximize her expected profit. The cumulative demand from the insider is denoted by $X:= X^B - X^S$ where $X^B$ and $X^S$ are cumulative buy and sell orders respectively. Since the insider observes the market price of the risky asset, she can back out the dynamics of noise orders, cf. discussions after Definition \ref{def: price rule}. Therefore the information structure of the insider $\F^I_t$ includes $\F^Z_t$ and $\sigma(\tilde{v})$, for any $t \in [0,1]$.

\item A competitive \emph{market maker} only observes the aggregation of the informed and noise trades, so he cannot distinguish between informed and noise trades. Given $Y:=X+Z$, the information of the market maker is $(\F^Y_t)_{t \in [0,1]}$ generated by $Y$ and satisfying the usual conditions. As the market maker is risk neutral, the competition will force him to set the market price as $\expec[\tilde{v}|\F^Y_t]$, $t \in [0,1]$.
\end{itemize}

In order to define equilibrium in the market, let us first describe admissible actions for the market maker and the insider. The market maker looks for a Markovian pricing mechanism, in which the price of the risky asset at time $t$ is set using cumulative order $Y_t$ and a pricing rule $p$.

\begin{defn}\label{def: price rule}
A function $p:\delta \Integer \times [0,1] \to \Real $ is a \emph{pricing rule} if
\begin{enumerate}[i)]
 \item $y \mapsto p(y, t)$ is strictly increasing for each $t \in [0,1)$;
 \item $\lim_{y\rightarrow -\infty} p(y,t) =v_1$ and $\lim_{y\rightarrow \infty} p(y,t) =v_N$ for each $t\in [0,1]$;
 \item $t\mapsto p(y,t)$ is continuous for each $y \in \delta\Integer$.
\end{enumerate}
\end{defn}

The monotonicity of $y\mapsto p(y, t)$ in i) is natural. It implies that the market price is higher whenever the demand is higher. Moreover, because of the monotonicity,  the insider fully observes the uninformed orders $Z$ by inverting the price process and subtracting her orders from the total orders. Item ii) means that the range of the pricing rule is wide enough to price in every possibility of fundamental value.

The insider trades to maximize her expected profit. Her admissible strategy is defined as follows:

\begin{defn}\label{def: insider ad}
The strategy $(X^B, X^S; \F^I)$ is \emph{admissible}, if
\begin{enumerate}[i)]
  \item $\F^I$ is a filtration satisfying the usual conditions and generated by $\sigma(\tv)$, $\F^Z$, and $\mathcal{H}$, where $(\mathcal{H}_t)_{t\in [0,1]}$ is a filtration independent of $\tv$ and $\F^Z$;
  \item $X^B$ and $X^S$, with $X^B_0 = X^S_0 =0$, are $\F^I$-adapted and integrable\footnote{That is, $\expec[X^B_1]$ and $\expec[X^S_1]$ are both finite.} increasing point processes with jump size $\delta$;
  \item the $(\F^I, \prob)$-dual predictable projections of $X^B$ and $X^S$ are absolutely continuous with respect to time, hence $X^B$ and $X^S$ admit $\F^I-$intensities $\theta^B$ and $\theta^S$, respectively;
  \item $\expec\bra{\int_0^1 |p(Y_t, t)| \,|dX^i_t-\delta \theta^i_t dt|}<\infty$, for $i\in \{B,S\}$ and the pricing rule $p$ fixed by the market maker. Here $|X^i-\int_0^\cdot \delta \theta^i dt|$ is the variation of the compensated point process.
\end{enumerate}
\end{defn}

This set of admissible strategies is similar to \cite[Definition 2.2]{Cetin-Xing}. Item i) assumes that the insider is allowed to possess additional information $\mathcal{H}$, independent of $\tv$ and $\F^Z$, which she uses to generate her mixed strategy. Item iv) implies $\delta \expec[\int_0^1 |p(Y_t, t)| \, \theta^i_t dt
]<\infty$, hence the expected profit of the insider is finite. Item ii) does not exclude the insider trading at the same time with noise traders. When the insider submits an order at the same time when an uniformed order arrives but in the opposite direction, assuming the market maker only observes the net demand implies that such pair of trades goes unnoticed by the market maker. This pair of opposite orders will be executed without a need for a market maker. Hence the market maker only knows the transaction when there is a need for him. Henceforth, when the insider makes a trade at the same time with an uninformed trader but in an opposite direction, we say the insider \emph{cancels} the noise trades. On the other hand, item ii) also allows the insider to trade at the same time with noise traders in the same direction. We call that the insider \emph{tops up} noise orders in this situation. However, the insider does not submit such orders in equilibrium, even when equilibrium is defined in a weak sense, cf. Remark \ref{rem: no top up} below.
The assumption that the insider is allowed to trade at the same time as noise traders is different from assumptions for Kyle model where insider's strategy is predictable. This additional freedom for insider is not the source for Theorem \ref{thm: nonexistence} below, which states optimizers for the insider's HJB equation do not correspond to the optimal strategy; see Remark \ref{rem: trading same time} below.

As described in the last paragraph, the insider's cumulative buy orders may consist of three components: $X^{B,B}$ arrives at different time than those of $Z^B$, $X^{B,T}$ arrives at the same time as some orders of $Z^B$, and $X^{B,S}$ cancels some orders of $Z^S$. Sell orders $X^S$ are defined analogously. Therefore $X^B= X^{B,B} + X^{B,T}+X^{B,S}$ and $X^S= X^{S,S} + X^{S,T}+ X^{S,B}$.

As mentioned earlier, the insider aims to maximize her expected profit. Given an admissible trading strategy $X=X^B -X^S$ the associated profit at time 1 of the insider is given by
\[
	\int^1_0 X_{t-}dp(Y_t,t) + (\tilde{v}-p(Y_1,1))X_1.
\]
The last term appears due to a potential discrepancy between the market price and the liquidation value. Since X is of finite variation and $X_0=0$, applying integration by parts rewrites the profit as
\begin{align*}
 &\int^1_0(\tilde{v}-p(Y_t,t))\,dX^B_t - \int^1_0(\tilde{v}-p(Y_t,t))\,dX^S_t \\
 &=\quad \int^1_0(\tilde{v}-p(Y_{t-} +\delta,t))\,dX^{B,B}_t + \int_0^1 (\tilde{v} - p(Y_{t-} + 2\delta, t)) \, dX^{B,T}+ \int^1_0(\tilde{v}-p(Y_{t-},t))\,dX^{B,S}_t  \\
 &\quad -\int^1_0(\tilde{v}-p(Y_{t-} -\delta,t))\,dX^{S,S}_t -\int_0^1 (\tilde{v} - p(Y_{t-} - 2\delta, t)) \, dX^{S,T}- \int^1_0(\tilde{v}-p(Y_{t-},t))\,dX^{S,B}_t,
\end{align*}
where $Y$ increases (resp. decreases) $\delta$ when $X^{B,B}$ (resp. $X^{S,S}$) jumps by $\delta$, $Y$ increases (resp. decreases) $2\delta$ when $X^{B,T}$ (resp. $X^{S,T}$) jumps at the same time with $Z^B$ (resp. $Z^S$), and $Y$ is unchanged when $X^{S,B}$ (resp. $X^{B,S}$) jumps at the same time with $Z^B$ (resp. $Z^S$). Define
\[
	a(y,t):=p(y+\delta,t) \quad   \text{ and } \quad b(y,t):=p(y-\delta,t),
\]
which can be viewed as ask and bid prices respectively.
Then the expected profit of the insider conditional on her information can be expressed as
\begin{equation}\label{eq: profit}
\begin{split}
& \expec\left[\int^1_0(\tilde{v}-a(Y_{t-},t))\,dX_t^{B,B} + \int_0^1 (\tilde{v} - a(Y_{t-} + \delta, t)) \, dX^{B,T}_t + \int^1_0(\tilde{v}-p(Y_{t-},t))\,dX_t^{B,S} \right.\\
&\quad \left.-\int^1_0(\tilde{v}-b(Y_{t-},t))\,dX_t^{S,S} -\int_0^1 (\tilde{v}-b(Y_{t-}-\delta, t))\, dX^{S,T}_t - \int^1_0(\tilde{v}-p(Y_{t-},t))\,dX_t^{S,B} \Big| \tilde{v}\right].
\end{split}
\end{equation}

Having described the market structure, an equilibrium between the market maker and the insider is defined as in \cite{Back-Baruch}:

\begin{defn}\label{def: GM eq}
A \emph{Glosten Milgrom equilibrium} is a quadruplet $(p,X^B,X^S, \F^I)$  such that
\begin{enumerate}[i)]
\item given $(X^B, X^S; \F^I)$, $p$ is a rational pricing rule, i.e., $p(Y_t,t)=\expec[\tilde{v}|\F^Y_t]$ for $t \in [0,1]$;
\item given $p$, $(X^B, X^S; \F^I)$ is an admissible strategy maximizing \eqref{eq: profit}.
\end{enumerate}
\end{defn}

When $N=2$, \cite{Cetin-Xing} establishes the existence of Glosten Milgrom equilibria. In equilibrium the pricing rule is
\begin{equation}\label{eq: p func GM}
 p(y,t) = \expec^{\prob^y}\bra{P(Z_{1-t})}, \quad (y,t) \in \delta \Integer \times [0,1].
\end{equation}
Here $\prob^y$ is a probability measure under which $Z$ is the difference of two independent Poisson processes and $\prob^y(Z_0=y)=1$. $P$ is a nondecreasing function such that $P(Z_1)$ has the same distribution as $\tv$. Moreover the optimal strategy of the insider are given by jump processes $X^{i,j}$, $i \in \{B,S\}$ and $j\in \{B,T,S\}$, with intensities $\delta\, \theta^{i,j}(Y_{t-}, t)$, $t\in [0,1]$. These intensities are deterministic functions of the state variable $Y$, hence this control strategy  is a \emph{feedback control} and it corresponds to optimizers of insider's HJB equation. However, when $N\geq 3$, Theorem \ref{thm: nonexistence} below shows that, given the pricing rule \eqref{eq: p func GM}, the optimal strategy \emph{does not} correspond to optimizers in the HJB equation, for some values of $\tv$. This result is surprising, because it is contrast to existing results in the Kyle and Glosten Milgrom equilibrium; cf. \cite{Kyle}, \cite{Back}, \cite{BP}, \cite{Back-Baruch}, \cite{Campi-Cetin}, \cite{Campi-Cetin-Danilova}, and \cite{Cetin-Xing}. This mismatch roots in the discrete state space of the demand process in the Glosten Milgrom model. The discrete state space yields different  bid and ask prices, which is contrast to the unique price in the Kyle model. See Remark \ref{rem: nonexistence} below for more discussion.

\subsection{Nonexistence of a feedback optimal control}\label{subsec: nonexistence}
To state aforementioned result, we introduce several additional notations. For each $\delta>0$, let $\Omega^\delta = \mathbb{D}([0,1], \delta \Integer)$ be the space of $\delta \Integer$-valued \cadlag\, functions on $[0,1]$ with coordinate process $Z^\delta$, $(\F^{Z, \delta}_t)_{t\in [0,1]}$ is the minimal right continuous and complete filtration generated by $Z^\delta$, and $\prob^\delta$ is the probability measure under which $Z^\delta$ is the difference of two independent Poisson processes starting from $0$ with the same jump size $\delta$ and intensity $\beta^\delta$. We denote by $\prob^{\delta, y}$ the probability measure under which $Z^\delta_0=y$ a.s.. Henceforth, the superscript $\delta$ indicates the trading size in the Glosten Milgrom model.

For the fundamental value $\tilde{v}^\delta$, let us first consider the following family of distributions.
\begin{ass}\label{ass: tv}
 Given $\tilde{v}^\delta$ of type \eqref{eq: v dist}, there exists a $\delta \Integer \cup \{-\infty, \infty\}-$valued strictly increasing sequence $(a^\delta_n)_{n=1, \cdots, N+1}$\footnote{When $N=\infty$, $N+1 =\infty$.} with $a_1^\delta = -\infty$, $a_{N+1}^\delta = \infty$, and $\cup_{n=1}^N [a^\delta_n, a^\delta_{n+1}) = \delta \Integer \cup\{-\infty\}$, such that
 \begin{equation}\label{def: v delta}
 \prob(\tv^\delta = v_n) = \prob^\delta \pare{Z^\delta_1 \in [a^\delta_n, a^\delta_{n+1})}, \quad n=1, \cdots, N.
\end{equation}
\end{ass}
For any $\tv$ with discrete distribution \eqref{eq: v dist}, Lemma \ref{lem: v delta}
below shows there exists a sequence $(\tilde{v}^\delta)_{\delta>0}$, each satisfies Assumption \ref{ass: tv} and converges to $\tilde{v}$ in law as $\delta \downarrow 0$. Therefore any $\tilde{v}$ of type \eqref{eq: v dist} can be approximated by a $\tilde{v}^\delta$ satisfying Assumption \ref{ass: tv}. Given $\tilde{v}^\delta$ satisfying Assumption \ref{ass: tv}, define
\begin{equation}\label{eq: h}
 h^\delta_n(y,t) := \prob^{\delta, y}\pare{Z^\delta_{1-t} \in [a_n^\delta, a_{n+1}^\delta)}, \quad y \in \delta \Integer, t\in [0,1], n\in \{1, \cdots, N\},
\end{equation}
and
\begin{equation}\label{eq: p func}
 p^\delta(y,t) := \sum_{n=1}^N v_n h^\delta_n(y,t)= \expec^{\delta, y}\bra{P(Z^\delta_{1-t})},
\end{equation}
where the expectation is taken under $\prob^{\delta, y}$ and
\begin{equation}\label{eq: P}
 P(y) = v_n, \quad \text{ when } y\in [a^\delta_n, a^\delta_{n+1}).
\end{equation}
Then \eqref{def: v delta} implies that $\tilde{v}^\delta$ and $P(Z_1^\delta)$ have the same distribution. If $p^\delta$ is chosen as the pricing rule, it has the same form as in \eqref{eq: p func GM}. Finally we impose a technical condition on $p^\delta$. This assumption is clearly satisfied when $N$ is finite.
\begin{ass}\label{ass: p poly}
 There exist positive constants $C$ and $n$ such that $|p^\delta(y,t)| \leq C(1+|y|^n)$ for any $(y,t)\in \delta \Integer \times [0,1]$.
\end{ass}

Given the pricing rule \eqref{eq: p func}, let us first study the insider's optimization problem and derive the associated HJB equation via a heuristic argument. In this derivation, the superscript $\delta$ is omitted to simplify notation. Definition \ref{def: insider ad} iii) implies that $X^{i,j}-\delta \int_0^\cdot \theta^{i,j}_r dr$ defines an $\F^I$-martingale for $i\in \{B, S\}$ and $j\in\{B,T,S\}$.
On the other hand,  Definition \ref{def: insider ad} iv) and \cite[Chapter I, T6]{Bremaud} combined imply that $\int_0^\cdot (\tv - p(Y_{r-}+\delta, r))(dX^{B,B}_r - \delta \theta^{B,B}_r dr)=\int_0^\cdot (\tv - p(Y_{r}, r)) (dX^{B,B}_r - \delta \theta^{B,B}_r dr)$ is an $\F^I$-martingale. Similar argument applied to other terms allows us to rewrite \eqref{eq: profit} as
\begin{align*}
\delta\expec\Big[&\int^1_0(\tilde{v}-p(Y_{r-}+\delta,r))\theta^{B,B}_r  \,dr+\int_0^1 (\tilde{v}-p(Y_{r-} +2\delta, r)) \theta^{B,T}_r\, dr + \int^1_0(\tilde{v}-p(Y_{r-},r))\theta^{B,S}_r \, dr \\
-&\int^1_0(\tilde{v}-p(Y_{r-}-\delta,r))\theta^{S,S}_r \, dr - \int_0^1 (\tilde{v} - p(Y_{r-}-2\delta, r)) \theta^{S,T}_r \, dr- \int^1_0(\tilde{v}-p(Y_{r-},r))\theta^{S,B}_r \, dr \Big| \tilde{v}\Big].
\end{align*}
This motivates us to define the following value function for the insider:
\begin{equation}\label{eq: value func}
\begin{split}
&V^\delta(\tv,y,t)  := \operatorname*{sup}_{\theta^{i,j};\,i\in\{B,S\}, j\in\{B,T,S\}} \\ \delta\, \expec \Big[&\int^1_t(\tilde{v}-p(Y_{r-}+\delta,r))\theta^{B,B}_r \,dr +\int_t^1 (\tv-p(Y_{r-}+2\delta, r)) \theta^{B,T}_r \, dr + \int^1_t(\tilde{v}-p(Y_{r-},r))\theta^{B,S}_r\, dr \\
-&\int^1_t(\tilde{v}-p(Y_{r-}-\delta,r))\theta^{S,S}_r\, dr -\int_t^1 (\tilde{v} - p(Y_{r-}-2\delta, r)) \theta^{S, T}_r\, dr - \int^1_t(\tilde{v}-p(Y_{r-},r))\theta^{S,B}_r\, dr \Big|Y_t=y,\tilde{v}\Big],
\end{split}
\end{equation}
for $\tv=\{v_1, \cdots, v_N\}$, $y\in \delta \Integer$, $t\in [0,1)$. The terminal value of $V^\delta$ is defined as $V^\delta(\tv, y,1) = \lim_{t\rightarrow 1} V^\delta(\tv, y, t)$ \footnote{Since the set of admissible control is unbounded, the HJB equation associated to \eqref{eq: value func} usually admits a \emph{boundary layer}, i.e., $\lim_{t\rightarrow 1} V^\delta(\tilde{v}, y, t)$ is not identically zero even if there is no terminal profit in \eqref{eq: profit}. Such phenomenon also shows up in Kyle model, see \cite{Back}.}. Lemma \ref{lem: V positive} and Proposition \ref{prop: U>V} below show that the optimization problem in \eqref{eq: value func} is well defined and nontrivial, i.e., $0<V^\delta<\infty$, for each $\delta>0$. Let us now derive the HJB equation which $V^\delta$ satisfies via a heuristic argument. Since positive (resp. negative) part of $Y$ is $Y^B:= X^{B,B} + X^{B,T}+Z^B - X^{S,B}$ (resp. $Y^{S} := X^{S,S} + X^{S,T}+ Z^S - X^{B,S}$). Hence $Y^B - \delta \int_0^\cdot (\beta- \theta^{S,B}_r-\theta^{B,T}_r) \,dr -\delta \int_0^\cdot \theta^{B,B}_r \, dr -2\delta \int_0^\cdot \theta^{B,T}_r\, dr$ (resp. $Y^S -\delta \int_0^\cdot (\beta -\theta^{B,S}_r - \theta^{S,T}_r)\, dr- \delta \int_0^\cdot \theta^{S,S}_r \, dr -2\delta \int_0^t \theta^{S,T}_r \, dr$) is an $\F^I$-martingale.\footnote{As discussed after Definition \ref{def: insider ad}, the set of jumps of $X^{B,S}$ and $X^{S,T}$ (resp. $X^{S,B}$ and $X^{B,T}$) arrive at the same time as some jumps of $Z^S$ (resp. $Z^B$), then we necessarily have $\theta^{B,S}+ \theta^{S,T} \leq \beta$ (resp. $\theta^{S, B}+ \theta^{B,T}\leq \beta$).} Then applying It\^{o}'s formula to $V^\delta(\tv, Y_r, r)$ and employing the standard dynamic programming arguments yield the following formal HJB equation for $V^\delta$:
\begin{equation}\label{eq: HJB}
-V_t(v_n, y,t) - H(v_n, y, t, V)=0, \quad n\in\{1, \cdots, N\}, (y,t)\in \delta\Integer\times[0,1),
\end{equation}
where the Hamilton $H$ is defined as (the $\tv$ argument is omitted in $H$ to simplify notation)
\begin{equation}\label{eq: Hamilton}
\begin{split}
H(v_n, y, t, V):=&
(V(y+\delta,t)-2V(y,t)+V(y-\delta,t))\beta \\
&+\operatorname*{sup}_{\theta^{B,B}\geq 0}\big[V(y+\delta,t)-V(y,t)+(v_n-p(y+\delta,t))\delta\big]\theta^{B,B}\\
&+\operatorname*{sup}_{\theta^{B,T}\geq 0}\big[V(y+2\delta,t)-V(y+\delta,t)+(v_n-p(y+2\delta,t))\delta\big]\theta^{B,T}\\
&+\operatorname*{sup}_{\theta^{B,S}\geq 0}\big[V(y,t)-V(y-\delta,t)+(v_n-p(y,t))\delta\big]\theta^{B,S} \\
&+\operatorname*{sup}_{\theta^{S,S}\geq 0}\big[V(y-\delta,t)-V(y,t)-(v_n-p(y-\delta,t))\delta\big]\theta^{S,S}\\
&+\operatorname*{sup}_{\theta^{S,T}\geq 0}\big[V(y-2\delta,t)-V(y-\delta,t)-(v_n-p(y-2\delta,t))\delta\big]\theta^{S,T}\\
&+\operatorname*{sup}_{\theta^{S,B}\geq 0}\big[V(y,t)-V(y+\delta,t)-(v_n-p(y,t))\delta\big]\theta^{S,B}.
\end{split}
\end{equation}
Optimizers $\theta^{i,j}$, $i\in \{B,S\}$ and $j\in \{B,T,S\}$, in \eqref{eq: Hamilton}, are deterministic functions of $v_n, y$ and $t$, hence they are of feedback form. They are expected to be the optimal control intensities for \eqref{eq: value func}. This is indeed the case in many existing results in Kyle model and Glosten Milgrom model (with $N=2$), compare \cite{Kyle}, \cite{Back}, \cite{BP}, \cite{Back-Baruch}, and \cite{Cetin-Xing}. However, when $N\geq 3$ in the Glosten Milgrom model, the following theorem shows any optimizers in \eqref{eq: Hamilton} are \emph{not} the optimal intensities when $\tv$ is neither $v_1$ nor $v_N$.

\begin{thm}\label{thm: nonexistence}
 Let $N\geq 3$ and $\tilde{v}^\delta$ satisfy Assumption \ref{ass: tv}. Let $p^\delta$ in \eqref{eq: p func} be the pricing rule and satisfy Assumption \ref{ass: p poly}. Then any optimizers $\theta^{i,j}(y,t)$, $i\in \{B,S\}, j\in\{B,T,S\}$ and  $(y, t)\in \delta \Integer\times [0,1)$, for \eqref{eq: Hamilton} are not the optimal strategy for \eqref{eq: value func} when $\tv^\delta=v_n$ for $1<n<N$.
\end{thm}

\begin{rem}\label{rem: nonexistence}
 When $\tilde{v}^\delta =v_1$ (resp. $v_N$), the insider knows the risky asset is always over-priced (resp. under-priced). Hence she always sells (resp. buys) in equilibrium. This situation is exactly the same as \cite{Cetin-Xing}. When $\tilde{v}^\delta$ is neither minimal nor maximal, let us briefly describe the proof of Theorem \ref{thm: nonexistence} here. To ensure \eqref{eq: HJB} to be wellposed, $H$ must be finite for all $(y,t)\in \delta \Integer \times [0,1)$. Hence
 \begin{equation}\label{eq: V bdd}
  (p(y,t)- v_n) \delta \leq V(y+\delta, t) - V(y,t) \leq (p(y+\delta, t)- v_n)\delta, \quad \text{ for all } (y,t)\in \delta \Integer \times [0,1),
 \end{equation}
 where the second inequality comes from the first three maximization in \eqref{eq: Hamilton} and the first inequality comes from the last three. Since $V>0$, $\theta^{i,j}\equiv 0$, $i\in \{B,S\}$ and $j\in \{B,T,S\}$, in \eqref{eq: Hamilton} does not correspond to the optimal strategy. Hence there must exist $(y_0, t_0)$ such that one inequality in \eqref{eq: V bdd}, say the first one, is an equality. However, in this case, the discrete state space forces the first inequality to be an equality for \emph{all} $(y,t)\in \delta \Integer \times [0,1)$, which implies the second inequality in \eqref{eq: V bdd} is strict for all $(y,t)$, due to $p(y+\delta, t) > p(y,t)$. Therefore the optimizers in the first three maximization in \eqref{eq: Hamilton} must be identically zero, which means the associated point process $X$ does not have positive jumps. On the other hand, the dynamic programming principle and the boundary layer of \eqref{eq: HJB} at $t=1$ force $Y_1= Z_1+X_1\in [a^\delta_n +\delta, a^\delta_{n+1}]$ a.s.. This can never happen when $X$ does not have positive jumps.
 Therefore, Theorem \ref{thm: nonexistence} is the joint effort of the discrete state space and the boundary layer of the HJB equation.
\end{rem}

\begin{rem}\label{rem: trading same time}
 The statement of Theorem \ref{thm: nonexistence} remains valid when the insider is prohibited from trading with noise traders at the same time; i.e., $X^{B,T}, X^{B,S}, X^{S,T}, X^{S,B}$ are all zero. In this case, the second, third, fifth and sixth maximization do not present in \eqref{eq: Hamilton}. However, the first and fourth maximization therein still lead to \eqref{eq: V bdd}. Hence the same argument as in the previous remark still applies.
\end{rem}

\begin{rem}\label{rem: relaxed control}
 Examples of control problems without optimal feedback control exist in literature of the optimal control theory, cf., e.g. \cite[Chapter 3, pp. 246]{Warga} and \cite[Example 1.1]{Lou}. In these cases, notion of \emph{relaxed control} is employed to prove the existence of a relaxed optimal control, cf. \cite{Lou} and references therein. For the insider's optimization problem, instead of $\{\theta: \delta \Integer\times [0,1]\rightarrow \Real_+\}$,  the control set can be relaxed to $\{\theta: \delta \Integer \times [0,1] \rightarrow \mathcal{M}^1(\Real_+)\}$, where $\mathcal{M}^1(\Real_+)$  is the set of all probability measures in $\Real_+$. It is interesting to investigate whether \eqref{eq: value func} admits an optimal control in this relaxed set. We leave this topic to future studies.
\end{rem}

\subsection{Asymptotic Glosten Milgrom equilibrium}
To establish equilibrium of Glosten Milgrom type when the risky asset $\tv$ has general discrete distribution \eqref{eq: v dist} with $N\geq 3$, we introduce a weak form of equilibrium in what follows. To motivate this definition, we recall the convergence of Glosten Milgrom equilibria as the order size decreasing to zero and intensity of noise trades increasing to infinity, cf. \cite[Theorem 3]{Back-Baruch} and \cite[Theorem 5.3]{Cetin-Xing}:

\begin{prop}\label{prop: GM conv}
 For any Bernoulli distributed $\tv$ (i.e. $N=2$ in \eqref{eq: v dist}), there exists a sequence of Bernoulli distributed random variables $\tv^{\delta}$ such that
 \begin{enumerate}[i)]
  \item $\tv^\delta$ converges to $\tv$ in law as $\delta \downarrow 0$;
  \item For each $\delta>0$, model with $\tv^\delta$ as the fundamental value of the risky asset admits a Glosten Milgrom equilibrium $(p^\delta, X^{B, \delta}, X^{S, \delta}, \F^{I, \delta})$;
  \item When the intensity of Poisson process is given by $\beta^\delta := (2\delta^2)^{-1}$, $X^{B,\delta}- X^{S, \delta} \wc X^0$, as $\delta \downarrow 0$,  where  $X^0$ is the optimal strategy in the Kyle model and $\wc$ represents the weak convergence of stochastic processes\footnote{Refer to \cite{Billingsley-conv} or \cite{Jacod-Shiryaev} for the definition of weak convergence of stochastic processes.}.
 \end{enumerate}
\end{prop}

This result motivates us to define the following weak form of Glosten Milgrom equilibrium:

\begin{defn}\label{def: Asy GM eq}
 For any $\tv$ with discrete distribution \eqref{eq: v dist}, an \emph{asymptotic Glosten Milgrom equilibrium} is a sequence $(\tv^\delta, p^\delta, X^{B,\delta}, X^{S, \delta}, \F^{I, \delta})_{\delta>0}$ such that
 \begin{enumerate}[i)]
  \item $\tv^\delta$ converges to $\tv$ in law as $\delta \downarrow 0$;
  \item For each $\delta>0$, given $(\tv^\delta, X^{B, \delta}, X^{S, \delta}, \F^{I, \delta})$ and set $Y^\delta:= Z^\delta + X^{B, \delta} - X^{S, \delta}$, $p^\delta$ is a rational pricing rule, i.e., $p^\delta(Y^\delta_t, t) = \expec[\tv^\delta \,|\, \F^{Y^\delta}_t]$ for $t\in [0,1]$;
  \item Given $(\tv^\delta, p^\delta)$ and $\beta^\delta=(2\delta^2)^{-1}$, let $\mathcal{J}^\delta(X^B, X^S)$ be insider's  expected profit associated to the admissible strategy $(X^B, X^S)$. Then
      \[
 \sup_{(X^B, X^S) \text{ admissible}} \mathcal{J}^\delta(X^B, X^S) - \mathcal{J}^\delta(X^{B, \delta}, X^{S, \delta})\rightarrow 0, \quad \text{as } \delta \downarrow 0.
 \]
 \end{enumerate}
\end{defn}

In the above definition, rationality of the pricing mechanism is not compromised. However optimality of the insider's strategy is not enforced. Instead, item iii) requires that, when the order size is small, the loss of insider's expected profit by employing the strategy $(X^{B, \delta}, X^{S, \delta}; \F^{\delta, I})$ is small, comparing to the optimal value. Moreover this discrepancy converges to zero when the order size vanishes. Therefore if the insider is willing to give up a small amount of expected profit, she can employ strategy $(X^{B, \delta}, X^{S, \delta}; \F^{I, \delta})$ to establish a suboptimal equilibrium. The following result establishes the existence of equilibrium in the above weak sense:

\begin{thm} \label{thm: main}
 Assume that $\tv$ satisfies \eqref{eq: v dist} with $N<\infty$. Then asymptotic Glosten-Milgrom equilibrium exists.
\end{thm}

In this asymptotic equilibrium, the pricing rule is given by \eqref{eq: p func}. When the order size is $\delta$, the insider employs the strategy $(X^{B, \delta}, X^{S, \delta}; \F^{I,\delta})$ with $\F^{I, \delta}$-intensities
      \begin{equation}\label{eq: X intensity}
      \begin{split}
       & \delta\beta^\delta\sum_{n=1}^N \indic_{\{\tv^\delta=v_n\}} \bra{\frac{h^\delta_n(Y^\delta_{t-} +\delta, t)}{h^\delta_n(Y^\delta_{t-}, t)}-1}_+  + \delta\beta^\delta\sum_{n=1}^N \indic_{\{\tv^\delta=v_n\}} \bra{\frac{h^\delta_n(Y^\delta_{t-}-\delta, t)}{h^\delta_n(Y^\delta_{t-}, t)}-1}_-, \\
       & \delta\beta^\delta\sum_{n=1}^N \indic_{\{\tv^\delta=v_n\}} \bra{\frac{h^\delta_n(Y^\delta_{t-} -\delta, t)}{h^\delta_n(Y^\delta_{t-}, t)}-1}_+  + \delta\beta^\delta\sum_{n=1}^N \indic_{\{\tv^\delta=v_n\}} \bra{\frac{h^\delta_n(Y^\delta_{t-}+\delta, t)}{h^\delta_n(Y^\delta_{t-}, t)}-1}_-,
      \end{split}
      \end{equation}
      respectively. In particular, when the fundamental value is $v_n$, the insider trades toward the middle level $m^\delta_n:= (a_n^\delta+ a_{n+1}^\delta-\delta)/2$ of the interval $[a^\delta_n, a^\delta_{n+1})$: when the total demand is less than $m^\delta_n$, the insider only places buy orders by either complementing noise buy orders or canceling some of noise sell orders, when the total demand is larger than $m^\delta_n$, the insider does exactly the opposite. More specifically, Lemma \ref{lem: hn property} below shows that $y\mapsto h^\delta_n(y,t)$ is strictly increasing when $y<m^\delta_n$ and strictly decreasing when $y>m^\delta_n$. Therefore, when $Y^\delta_{t-} < m^\delta_n$,  \eqref{eq: X intensity} implies that: $X^{B,B, \delta}$ has intensity $\frac{1}{2\delta} \pare{\frac{h^\delta_n(Y^\delta_{t-} +\delta, t)}{h_n^\delta(Y^\delta_{t-}, t)}-1}$, $X^{B,S,\delta}$ has intensity $\frac{1}{2\delta} \pare{1-\frac{h^\delta_n(Y^\delta_{t-}-\delta,t)}{h^\delta_n(Y^\delta_{t-}, t)}}$, meanwhile intensities of $X^{S,S, \delta}$ and $X^{S,B, \delta}$ are both zero. When $Y^\delta_{r-}>m^\delta_n$, intensities can be read out from \eqref{eq: X intensity} similarly.
      Even though Theorem \ref{thm: nonexistence} remains valid when the insider is prohibited from trading at the same time with noise traders, the strategy constructed above depends on the possibility of canceling orders. However, in this strategy, the insider never tops up noise orders, i.e., $X^{B,T}=X^{S,T}\equiv 0$. This allows the market maker to employ a rational pricing mechanism so that Definition \ref{def: Asy GM eq} ii) is satisfied, cf. Remark \ref{rem: no top up} below.

      The processes $(X^{B,\delta}, X^{S, \delta}; \F^{I, \delta})$ with intensities \eqref{eq: X intensity} will be constructed explicitly in Section \ref{sec: point process bridge}. The insider employs a sequence of independent random variables with uniform distribution on $[0,1]$ to construct her mixed strategy. This sequence of random variables are also independent of $Z^\delta$ and $\tv^\delta$. This construction is a natural extension of \cite{Cetin-Xing}. In this construction, whenever a noise order arrives, the insider uses a uniform distributed random variable to decide whether or not submitting an opposite canceling order. Hence this strategy is adapted to insider's filtration, rather than predictable as in the Kyle model. Such a canceling strategy is called \emph{input regulation} and has been studied extensively in the queueing theory literature, see eg. \cite[Chapter VII, Section 3]{Bremaud}.

      When the fundamental value is $v_n$ and the insider follows the aforementioned strategy, the total demand at time $1$ will end up in the interval $[a^\delta_n, a^\delta_{n+1})$. Therefore the insider's private information is fully, albeit gradually, revealed to the public so that the trading price does not jump when the fundamental value is announced. On the other hand, the total demand, in its own filtration, has the same distribution of the demand from noise traders, i.e., the insider is able to hide her trades among the noise trades. This is another manifestation of \emph{inconspicuous trading theorem} commonly observed in the insider trading literature (cf. e.g., \cite{Kyle}, \cite{Back}, \cite{BP}, etc.).

      The insider's strategy discussed above is of feedback form. The insider can determine her trades only using the current total cumulative demand (and some additional randomness coming from the sequence of iid uniform distributed random variables which are also independent of the fundamental value and the noise trades). Even though this strategy is not optimal, its associated expected profit is close to the optimal value when the order size is small. Moreover the discrepancy converges to zero as the order size diminishes.

      The following numeric example illustrates the convergence of the upper bound for insider's expected profit loss as the order size decreases to zero. In this example, $\tv$ takes values in $\{1,2,3\}$ with probability $0.55$, $0.35$, and $0.1$, respectively. The expected profit in Kyle-Back equilibrium is $0.512$. Compared to this, the following figure shows that the loss to insider's expected profit is small.
\begin{figure} [ht]
\centering
\includegraphics[scale=0.8]{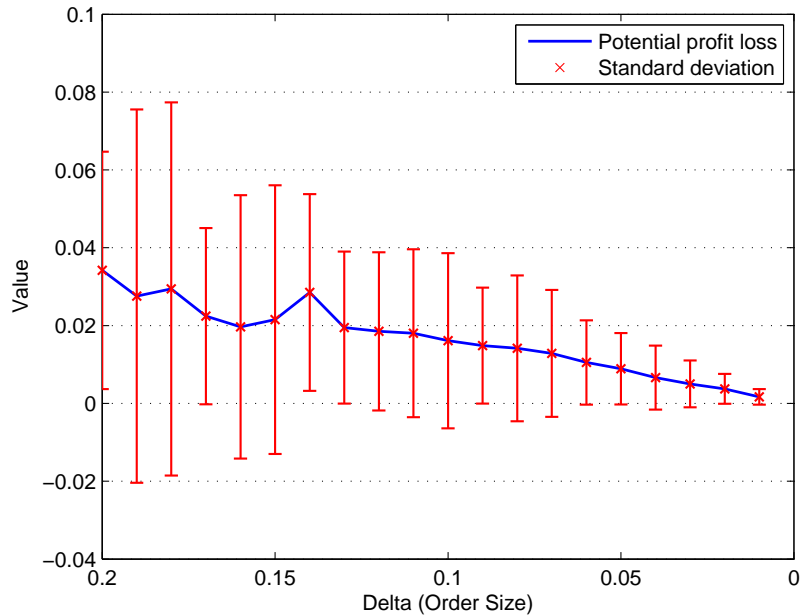}
\caption{The mean and standard deviation of the upper bound for insider's expected profit loss. The figure is generated by Monte Carlo simulation with $10^5$ paths.}
\label{figure}
\end{figure}

Finally, similar to Proposition \ref{prop: GM conv} iii), insider's net order in the asymptotic Glosten Milgrom equilibrium converges to the optimal strategy in the Kyle model as the order size decreases to zero.

\begin{thm}\label{thm: str conv}
 Let $(X^{B,\delta}, X^{S, \delta}, \F^{I, \delta})_{\delta>0}$ be the sequence of insider's strategy in Theorem \ref{thm: main}. Then
 \[
  X^{B,\delta}- X^{S, \delta} \wc X^0 \quad \text{ as }\delta \downarrow 0,
 \]
 where $X^0$ is the optimal strategy in Kyle model.
\end{thm}

\section{Optimizers in the HJB equation are not optimal control} \label{sec: nonexistence}

Theorem \ref{thm: nonexistence} will be proved in this section. Let us first make the heuristic argument for the HJB equation rigorous by using the dynamic programming principle and standard arguments for viscosity solutions. To this end, recall the domain of Hamilton:
\[
 \operatorname*{dom}(H):= \left\{(v_n, y, t, V)\in\{v_1, \cdots, v_N\} \times \delta \Integer \times [0,1] \times \Real-\text{valued functions} \such H(v_n, y, t, V)<\infty\right\}.
\]
Observe that control variables for \eqref{eq: Hamilton} are chosen in $[0,\infty)$. Hence $(v_n, y, t, V) \in \text{dom}(H)$ if
\begin{align}
 & V(y+\delta, t) - V(y,t) + (v_n - p(y+\delta, t)) \delta \leq 0,\label{eq: buy ineq}\\
 & V(y-\delta, t) - V(y,t) - (v_n - p(y-\delta, t))\delta \leq 0. \label{eq: sell ineq}
\end{align}
Moreover, when $(v_n, y, t, V)\in \text{dom}(H)$, the Hamilton is reduced to
\begin{equation}\label{eq: H red}
 H(v_n, y,t,V) = (V(y+\delta,t)-2V(y,t)+V(y-\delta,t))\beta.
\end{equation}
Hence \eqref{eq: HJB} reads
\begin{equation}\label{eq: red HJB}
 -V_t - (V(y+\delta,t)-2V(y,t)+V(y-\delta,t))\beta =0 \quad \text{ in } \text{dom}(H).
\end{equation}

\begin{prop}\label{prop: vis soln}
 The following statements hold for $V^\delta$, $\delta>0$:
 \begin{enumerate}[i)]
  \item $V^\delta$ is a viscosity solution of \eqref{eq: HJB};
  \item $(v_n, y, t, V^\delta)\in \text{dom}(H)$ for any $n\in\{1, \cdots, N\}$ and $(y, t)\in \delta \Integer\times [0,1)$. Hence $V^\delta$ satisfies \eqref{eq: buy ineq}, \eqref{eq: sell ineq}, and is a viscosity solution of \eqref{eq: red HJB};
  \item $t\mapsto V^\delta(y,t)$ is continuous on $[0,1]$;
  \item $V^\delta(y,t) = \expec^{\prob^{\delta, y}}\bra{V^\delta(Z_{s-t}, s)}$ for any $y\in \delta \Integer$,  and $0\leq t\leq s\leq 1$.
 \end{enumerate}
\end{prop}
The proof is postponed to Appendix \ref{app: vis soln} where the dynamic programming principle together with the definition of viscosity solutions are recalled. The proof of Theorem \ref{thm: nonexistence} also requires the following result.

\begin{lem}\label{lem: V positive}
 For any $\delta>0$, $n \in \{1, \cdots, N\}$, and $(y,t)\in \delta\Integer \times [0,1)$, $V^\delta(v_n,y,t)>0$.
\end{lem}
\begin{proof}
Without loss of generality, we fix $\delta =1$, $\tv=v_n$ for some $n\in\{1, \cdots, N\}$, and $(y,t)=(0,0)$. The superscript $\delta$ is omitted throughout this proof. When $n>1$, let us construct a strategy where the insider buys once the asset is under-priced. Consider
\[
 \tau:= \inf\{r\,:\, p(Z_{r-}+1, r)<v_n)\}\wedge 1 \quad \text{ and } \quad \sigma:= \inf\{r> \tau\,:\, \Delta Y_r \neq 0\} \wedge 1.
\]
Here $\tau$ is the first time that the asset is under-priced and $\sigma$ is the arrival time of the first order after $\tau$. The insider employs a strategy with intensity $\theta^{B,B}_r = \indic_{\{\tau \leq r\leq \sigma\}}$ and all other intensities zero. Then the associated expected profit is
\[
 \expec\bra{\int_0^1 \pare{v_n - a(Y_{r-}, r)} \indic_{\{\tau \leq r\leq \sigma\}} dr} = \expec\bra{\int_\tau^\sigma \pare{v_n - p(Z_{r-}+1, r)} dr} >0,
\]
where the inequality follows from the definition of $\tau$ and the fact that $\prob(\tau<1)>0$ due to Definition \ref{def: price rule} ii). When $n=1$, set $\tau:= \inf\{t\,:\, p(Z_{t-}-1, t)>v_1\}\wedge 1$ and $\theta^{S,S}_t = \indic_{\{\tau\leq t\leq \sigma\}}$. Argument similar as above shows that this selling strategy also leads to positive expected profit. Therefore, in both cases, $V>0$ is verified.
\end{proof}

\begin{proof}[Proof of Theorem \ref{thm: nonexistence}]
Without loss of generality, we set $\delta =1$ and omit the superscript $\delta$ throughout the proof.

\vspace{2mm}
\noindent \underline{Step 1:} For any $n\in \{1, \cdots, N\}$, either one of the following situations holds:
\begin{itemize}
 \item \eqref{eq: buy ineq} holds as an equality and \eqref{eq: sell ineq} is a strict inequality at all $(y,t)\in \Integer \times [0,1)$;
 \item \eqref{eq: sell ineq} holds as an equality and \eqref{eq: buy ineq} is a strict inequality at all $(y,t)\in \Integer \times [0,1)$.
\end{itemize}
To prove the assertion, observe from \eqref{eq: buy ineq} and \eqref{eq: sell ineq} that
\[
 p(y,t) -v_n \leq V(y+1, t) - V(y,t) \leq p(y+1, t) - v_n, \quad (y, t)\in \Integer\times [0,1).
\]
Since $y\mapsto p(y,t)$ is strictly increasing for any $t\in[0,1)$, there exists $\eta(y,t) \in [0,1]$ such that
\[
 V(y+1,t) - V(y,t) = p(y,t) + \eta(y,t) \pare{p(y+1, t) - p(y,t)} - v_n, \quad (y, t)\in \Integer\times [0,1).
\]

Assume that either \eqref{eq: buy ineq} or \eqref{eq: sell ineq} holds as an equality at some point. If such assumption fails, both inequalities in \eqref{eq: buy ineq} and \eqref{eq: sell ineq} are strict at all points in $\Integer \times [0,1)$. Then all optimizers in \eqref{eq: Hamilton} are identically zero, with the associated expected profit zero. Since $V>0$, cf. Lemma \ref{lem: V positive}, these trivial optimizers  are not optimal strategies for \eqref{eq: value func}. Hence the statement of the theorem is already confirmed in this trivial situation. Let us now assume \eqref{eq: sell ineq} holds as an equality at $(y_0+1, t_0)$, we will show \eqref{eq: sell ineq} is an identity. On the other hand, combining the identity in \eqref{eq: sell ineq} and the strict monotonicity of $y\mapsto p(y,t)$, we obtain
\[
 V(y+1, t) - V(y,t) = p(y, t) - v_n < p(y+1, t) - v_n, \quad (y,t)\in \Integer\times[0,1),
\]
hence the inequality \eqref{eq: buy ineq} is always strict. The other situation where \eqref{eq: buy ineq} is an identity and \eqref{eq: sell ineq} is strict can be proved analogously.

Since \eqref{eq: sell ineq} holds as an equality at $(y_0+1, t_0)$, then, for any $s\in (t_0, 1)$,
\begin{equation*}
 \expec^{y_0}\bra{p(Z_{s-t_0}, s)} - v_n = p(y_0, t_0) - v_n = V(y_0+1, t_0) - V(y_0, t_0) = \expec^{y_0}\bra{V(Z_{s-t_0}+1, s) - V(Z_{s-t_0}, s)},
\end{equation*}
where the first identity follows from \eqref{eq: p func} and the Markov property of $Z$, the third identity is obtained after applying Proposition \ref{prop: vis soln} iv) twice. On the other hand, the definition of $\eta(y, t)$ yields
\begin{equation*}
 \expec^{y_0}\bra{V(Z_{s-t_0}+1, s) - V(Z_{s-t_0}, s)} = \expec^{y_0}\bra{p(Z_{s-t_0}, s) + \eta(Z_{s-t_0}, s)\pare{p(Z_{s-t_0}+1, s) - p(Z_{s-t_0}, s)}} -v_n.
\end{equation*}
The last two identities combined imply
\begin{equation}\label{eq: exp eta 0}
 \expec^{y_0} \bra{\eta(Z_{s-t_0}, s) \pare{p(Z_{s-t_0}+1, s) - p(Z_{s-t_0}, s)}} =0.
\end{equation}
Recall that $\eta\geq 0$, $p(\cdot+1, s)-p(\cdot, s)>0$ for any $s<1$, and the distribution of $Z_{s_0-t}$ has positive mass on each point in $\Integer$. We then conclude from \eqref{eq: exp eta 0} that $\eta(y,s) =0$ for any $y\in \Integer$. Since $s$ is arbitrarily chosen,
\begin{equation}\label{eq: eta 0 s}
 \eta(y,s) =0, \quad \text{for any } y\in \Integer, t_0<s<1.
\end{equation}
Now fix $s$, the previous identity yields, for any $t<s$ and $y\in \Integer$,
\[
 V(y+1, t) -V(y,t) = \expec^y\bra{V(Z_{s-t}+1, s) - V(Z_{s-t}, s)} = \expec^y\bra{p(Z_{s-t}, s)} - v_n = p(y,t) - v_n,
\]
where Proposition \ref{prop: vis soln} iv) is applied twice again to obtain the first identity. Therefore $\eta(y,t) =0$ for any $y\in \Integer$ and $t\leq s$, which combined with \eqref{eq: eta 0 s}, implies \eqref{eq: sell ineq} is an identity.

\vspace{2mm}

\noindent \underline{Step 2:} Fix $1<n<N$. When \eqref{eq: sell ineq} is an identity, any optimizers in \eqref{eq: Hamilton} are shown not to be the optimal strategy for \eqref{eq: value func}. When \eqref{eq: buy ineq} is an identity, a similar argument leads to the same conclusion. Combined with the result in Step 1, the statement of the theorem is confirmed.

When \eqref{eq: sell ineq} is an identity, sending $t\rightarrow 1$, $V(y,1)$, defined as $\lim_{t\rightarrow 1} V(y,t)$, satisfies
\[
 V(y-1,1) - V(y,1) = v_n-P(y-1).
\]
The previous identity and \eqref{eq: P} combined imply that $V(y,1)$ is strictly decreasing when $y< a_n+1$, constant when $y\in [a_n+1, a_{n+1}+1)$, and strictly increasing when $y\geq a_{n+1}+1$. Thus $y\mapsto V(y,1)$ attains its minimum value when $y\in [a_n+1, a_{n+1}]$. Let $(\hat{X}^B, \hat{X}^S)$ be the point processes whose $\F^I$-intensities are optimizers $\hat{\theta}^{i,j}$, $i\in \{B,S\}$ and $j\in \{B,T,S\}$, in \eqref{eq: Hamilton}, and set $\hat{Y}= Z + \hat{X}^B- \hat{X}^S$. Assuming that $(\hat{X}^B, \hat{X}^S)$ is the optimal strategy for \eqref{eq: value func}, DPP i) in Appendix \ref{app: vis soln} implies
\[
     \begin{split}
      &V(y,t) \\
      \geq &\expec^{y,t} \left[V(\hat{Y}_1, 1)\right.\\
      & + \int_t^1 (v_n-p(\hat{Y}_{r-}+1,r))d\hat{X}_r^{B,B} + \int_t^1 (v_n - p(\hat{Y}_{r-}+2,r)) d\hat{X}^{B,T}_r + \int_t^1 (v_n-p(\hat{Y}_{r-},r))d\hat{X}_r^{B,S}\\
      &\left.-\int_t^1 (v_n-p(\hat{Y}_{r-}-1,r))d\hat{X}_r^{S,S} -\int_t^1 (v_n - p(\hat{Y}_{r-}-2, r)) d\hat{X}^{S,T}_r - \int_t^1 (v_n-p(\hat{Y}_{r-},r))d\hat{X}_r^{S,B} \right],
     \end{split}
\]
where the expectation is taken under $\prob^{y,t}$ with $\prob^{y,t}(\hat{Y}_t = y)=1$.
However, the value function $V(y,t)$ is exactly the expected profit when the insider employs the optimal strategy $(\hat{X}^B, \hat{X}^S)$. Therefore, the previous identity yields
\[
 \expec^{y,t}[V(\hat{Y}_1, 1)] =0.
\]
Recall that $V(\cdot, 1)$, as limit of positive functions, is nonnegative, and it achieves the minimum at $[a_n+1, a_{n+1}]$. The previous identity implies  $V(y,1) =0$ when $y\in [a_n+1, a_{n+1}]$ and
\begin{equation}\label{eq: Y_1 terminal}
 \hat{Y}_1 \in [a_n+1, a_{n+1}], \quad \prob^{y,t}-a.s..
\end{equation}

However, when \eqref{eq: sell ineq} is an identity and \eqref{eq: buy ineq} is a strict inequality, any optimizer of \eqref{eq: Hamilton} satisfies $\hat{\theta}^{B,B}= \hat{\theta}^{B,S}\equiv 0$, i.e., $\hat{X}^{B}\equiv 0$. Therefore, $\hat{Y} = Z^B - Z^S - \hat{X}^S$ with only negative controlled jumps from $\hat{X}^S$ cannot compensate $Z^S$ to satisfy \eqref{eq: Y_1 terminal}, where $[a_n+1, a_{n+1}]$ is a finite interval in $\Integer$ when $1<n<N$.
\end{proof}

\section{A suboptimal strategy}\label{sec: suboptimal}

We start to prepare the proof of Theorem \ref{thm: main} from this section.

\begin{center}
\textit{For the rest of the paper, $N<\infty$, assumed in Theorem \ref{thm: main}, is enforced unless stated otherwise.}
\end{center}

In this section we are going to characterize a suboptimal strategy of feedback form in the Glosten Milgrom model with order size $\delta$, such that the pricing rule \eqref{eq: p func} is rational. To simplify presentation, we will take $\delta =1$, hence omit all superscript $\delta$, throughout this section. Scaling all processes by $\delta$ gives the desired processes when the order size is $\delta$.

The following standing assumption on distribution of $\tv$ will be enforced throughout this section:
\begin{ass}\label{ass: middle level}
 There exists a strictly increasing sequence $(a_n)_{n=1, \cdots, N+1}$ such that
 \begin{enumerate}[i)]
  \item $a_n\in \Integer\cup \{-\infty, \infty\}$, $a_1=-\infty$, $a_{N+1}=\infty$, and $\cup_{n=1}^{N} [a_n, a_{n+1}) = \Integer\cup\{-\infty\}$;
  \item $\prob(Z_1 \in [a_n, a_{n+1})) = \prob(\tv=v_n)$, $n=1, \cdots, N$;
  \item The middle level $m_n=(a_n + a_{n+1}-1)/2$ of the interval $[a_n, a_{n+1})$ is not an integer.
 \end{enumerate}
\end{ass}
Item i) and ii) have already been assumed in Assumption \ref{ass: tv}. Item iii) is a technical assumption which facilities the construction of the suboptimal strategy. In the next section, when an arbitrary $\tv$ of distribution \eqref{eq: v dist} is considered and the order size $\delta$ converges to zero, a sequence  $(a_n^\delta)_{n=1, \cdots, N+1, \delta>0}$ together with a sequence of random variables $(\tv^\delta)_{\delta>0}$ will be constructed, such that Assumption \ref{ass: middle level} is satisfied for each $\delta$ and $\tv^\delta$ converges to $\tv$ in law. To simplify notation, we denote by $\dm:= \lfloor (a_n+ a_{n+1}-1)/2\rfloor$ the largest integer smaller than $m_n$ and by $\um:= \lceil (a_n+ a_{n+1}-1)/2 \rceil$ the smallest integer larger than $m_n$. Assumption \ref{ass: middle level} iii) implies $a_n \leq \dm < m_n < \um < a_{n+1}$ and $\um - \dm =1$ when both $a_n$ and $a_{n+1}$ are finite.

Let us now define a function $U$, which relates to the expected profit of a suboptimal strategy and also dominates the value function $V$. First the Markov property $Z$ implies that $p$ is continuously differentiable in the time variable and satisfies\footnote{This follows from the same argument as in \cite[Footnote 4]{Cetin-Xing}.}
\begin{equation}\label{eq: p pde}
\begin{split}
 & p_t + \pare{p(y+1, t) - 2p(y,t) + p(y-1, t)} \beta =0, \quad (y,t)\in \Integer \times [0,1),\\
 & p(y,1) = P(y).
\end{split}
\end{equation}
Define
\begin{equation}\label{eq: U 1}
 U(v_n, y, 1) := \sum_{j=y}^{a_n-1} (v_n - A(j)) \, \indic_{\{y\leq \dm\}} + \sum_{j=a_{n+1}}^y (B(j)- v_n) \, \indic_{\{y\geq \um\}}, \quad y\in \Integer, 1\leq n\leq N,
\end{equation}
where $A(y):= P(y+1)$ and $B(y):= P(y-1)$ can be considered as ask and bid pricing functions right before time $1$.  Since $(v_n)_{n=1, \cdots, N}$ is increasing, $U(\cdot, \cdot, 1)$ is nonnegative and
\begin{equation}\label{eq: U=0}
 U(v_n, y, 1) =0 \quad \Longleftrightarrow \quad y\in [a_n-1, a_{n+1}+1).
\end{equation}
Given $U(\cdot, \cdot, 1)$ as above, $U$ is extended to $t\in [0,1)$ as follows:
\begin{align}
 U(v_n, y, t) &:= U(v_n, y, 1) + \beta \int_t^1 \pare{p(y,r) - p(y-1, r)} dr, \quad y \geq \um,\label{eq: U um}\\
 U(v_n, y, t) &:= U(v_n, y, 1) + \beta \int_t^1 \pare{p(y+1, r) - p(y, r)} dr, \quad y\leq \dm, \label{eq: U dm}
\end{align}
for $t\in [0,1)$ and $n=1, \cdots, N$. Since $N$ is finite, $p$ is bounded, hence $U$ takes finite value.

\begin{prop}\label{prop: ep<U-L}
 Let Assumption \ref{ass: middle level} hold. Suppose that the market maker chooses $p$ in \eqref{eq: p func} as the pricing rule. Then for any insider's admissible strategy $(X^B, X^S; \F^I)$, with $\F^I$-intensities $\theta^{i,j}, i\in \{B,S\}$ and $j\in\{B,T,S\}$, the associated expected profit function $\mathcal{J}(v_n, y, t; X^B, X^S)$ satisfies
 \begin{equation}\label{eq: ep<U-L}
  \mathcal{J}(v_n, y,t; X^B, X^S) \leq U(v_n, y, t) - L(v_n, y, t), \quad n\in\{1, \cdots, N\}, (y,t)\in \Integer \times [0,1].
 \end{equation}
 where
 \begin{equation}\label{eq: L}
 \begin{split}
 &L(v_n, y, t) \\
 := & \quad \expec^y \bra{\left.\int_t^1 \pare{v_n- p(\dm, r)} \bra{\pare{\beta - \theta^{B, S}_r + \theta^{S,S}_r} \, \indic_{\{Y_{r-} = \um\}} + \theta^{S,T}_r \, \indic_{\{Y_{r-}=\um+1\}}}dr \right| \tv= v_n} \\
  & - \expec^y \bra{\left.\int_t^1 \pare{v_n- p(\um, r)} \bra{\pare{\beta - \theta^{S,B}_r + \theta^{B,B}_r} \, \indic_{\{Y_{r-} = \dm\}} + \theta^{B,T}_r \, \indic_{\{Y_{r-}=\dm-1\}}}dr \right| \tv= v_n}.
 \end{split}
 \end{equation}
 Moreover \eqref{eq: ep<U-L} is an identity when the following conditions are satisfied:
 \begin{enumerate}[i)]
  \item $Y_1 \in [a_n-1, a_{n+1}+1)$ a.s. when $\tv= v_n$;
  \item $X^{S,S}_t = X^{S,B}_t \equiv 0$ when $Y_{t-} \leq \dm$, $X^{B,B}_t = X^{B,S}_t \equiv 0$ when $Y_{t-}\geq \um$, $\theta^{B,T}\equiv 0$ when $y\geq \dm$, and $\theta^{S,T}\equiv 0$ when $y\leq \um$.
 \end{enumerate}
\end{prop}
Before proving this result, let us derive equations that $U$ satisfies. The following result shows that $U$ satisfies \eqref{eq: red HJB} except when $y=\um$ and $y= \dm$, and $U$ satisfies the identity in either \eqref{eq: buy ineq} or \eqref{eq: sell ineq} depending on whether $y\leq \dm$ or $y\geq \um$.

\begin{lem}\label{lem: U eqn}
 The function $U$ satisfies the following equations: (Here $\tv=v_n$ is fixed and the dependence on $\tv$ is omitted in $U$.)
 \begin{align}
  & U_t + \pare{U(y+1, t) - 2U(y,t) + U(y-1, t)} \beta =0, & y>\um \text{ or } y<\dm, \label{eq: U har}\\
  & U_t + \pare{U(y+1, t) - 2U(y,t) + U(y-1, t)} \beta = (p(\dm, t) - v_n) \beta, & y= \um, \label{eq: U eq um}\\
  & U_t + \pare{U(y+1, t) - 2U(y,t) + U(y-1, t)} \beta = (v_n - p(\um, t)) \beta, & y= \dm, \label{eq: U eq dm}\\
  & U(y, t) - U(y+1,t) - (v_n - p(y,t))=0, & y\geq \um, \label{eq: U eqn >um}\\
  & U(y, t) - U(y-1,t) + (v_n - p(y,t))=0, & y\leq \dm, \label{eq: U eq <dm}
 \end{align}
\end{lem}

\begin{proof}
 We will only verify these equations when $y\geq \um$. The remaining equations can be proved similarly. First \eqref{eq: U 1} implies
 \[
  U(y+1, 1) - U(y,1) = B(y+1) - v_n = P(y) - v_n, \quad y\geq \um.
 \]
 Combining the previous identity with \eqref{eq: U um},
 \begin{align*}
  U(y+1, t) - U(y,t) &= U(y+1, 1) - U(y,1) + \beta \int_t^1 \pare{p(y+1, r) - 2p(y,r) + p(y-1, r)} dr \\
  &= p(y,t) -v_n,
 \end{align*}
 where \eqref{eq: p pde} is used to obtain the second identity. This verifies \eqref{eq: U eqn >um}. When $y> \um$, summing up \eqref{eq: U eqn >um} at $y$ and $y+1$, and taking time derivative in \eqref{eq: U um}, yield
 \begin{align*}
 & U_t  +\pare{U(y+1, t) - 2 U(y,t) + U(y-1, t)} \beta \\
 &= -\beta(p(y,t) - p(y-1, t)) + \beta (p(y,t) - p(y-1, t))\\
 &=0,
 \end{align*}
 which confirms \eqref{eq: U har} when $y>\um$. When $y= \um$, observe from \eqref{eq: U 1}, \eqref{eq: U um} and \eqref{eq: U dm} that $U(\um, \cdot) = U(\dm, \cdot)$. Then
 \begin{align*}
 & U_t + \pare{U(y+1, t) - 2U(y,t) + U(y-1, t)} \beta \\
 &= -\beta\pare{p(\um, t) - p(\dm, t)} + \beta\pare{U(\um+1, t) - U(\um, t)}\\
 &= -\beta \pare{p(\um, t) - p(\dm, t)} - \beta \pare{v_n - p(\um, t)}\\
 &= \beta\pare{p(\dm, t) - v_n},
 \end{align*}
 where the second identity follows from \eqref{eq: U eqn >um}.
\end{proof}

\begin{proof}[Proof of Proposition \ref{prop: ep<U-L}]
 Throughout the proof the $\tv=v_n$ is fixed and the dependence on $\tv$ is omitted in $U$. Let $Y^B= Z^B + X^{B,B} + X^{B,T}- X^{S, B}$ and $Y^S= Z^S + X^{S,S}+ X^{S,T} - X^{B,S}$ be positive and negative parts of $Y$ respectively. Then $Y^B - \int_0^\cdot (\beta - \theta^{S, B}_r -\theta^{B,T}_r) dr -\int_0^\cdot \theta^{B,B}_r dr -2\int_0^\cdot \theta^{B,T}_r dr$ and $Y^S- \int_0^\cdot (\beta - \theta^{B,S}_r - \theta^{S,T}_r) dr - \int_0^\cdot \theta^{S,S}_r dr - 2\int_0^\cdot \theta^{S,T}_r dr$ are $\F^I$-martingales. Applying It\^{o}'s formula to $U(Y_\cdot, \cdot)$, we obtain
 \begin{equation}\label{eq: U ito}
 \begin{split}
  &U(Y_1, 1) \\
  = &U(y,t) + \int_t^1 U_t (Y_{r-}, r) dr \\
  & + \int_t^1 \bra{U(Y_{r}, r) - U(Y_{r-}, r)}  dY^B_r + \int_t^1 \bra{U(Y_{r}, r) - U(Y_{r-}, r)} dY^S_r\\
  =& U(y,t) + \int_t^1 \bra{U_t(Y_{r-}, r) +\pare{U(Y_{r-}+1, r) - 2U(Y_{r-}, r) + U(Y_{r-}-1, r)}\beta}  dr\\
  &+ \int_t^1 \bra{U(Y_{r-} +1, r) - U(Y_{r-}, r)} \pare{\theta^{B,B}_r- \theta^{S,B}_r} dr\\
  &+ \int_t^1 \bra{U(Y_{r-}+2,r) - U(Y_{r-}+1,r)} \theta^{B,T}_r dr\\
  &+\int_t^1 \bra{U(Y_{r-}-1, r) - U(Y_{r-}, r)} \pare{\theta^{S,S}_r - \theta^{B,S}_r} dr \\
  &+ \int_t^1 \bra{U(Y_{r-}-2, r) - U(Y_{r-}-1, r)} \theta^{S,T}_r dr + M_1 - M_t,
 \end{split}
 \end{equation}
 where
 \begin{align*}
  M= &\quad\int_0^\cdot \bra{U(Y_{r}, r) - U(Y_{r-}, r)} d\pare{Y^B_r - \int_0^r \pare{\beta-\theta^{S,B}_u+ \theta^{B,B}_u + \theta^{B,T}_u} du}\\
  & + \int_0^\cdot \bra{U(Y_{r}, r) - U(Y_{r-}, r)} d\pare{Y^S_r - \int_0^r \pare{\beta- \theta^{B,S}_u + \theta^{S,S}_u+\theta^{S,T}_u} du}.
 \end{align*}
 Since \eqref{eq: U eqn >um} and \eqref{eq: U eq <dm} imply $U(y+1, t)- U(y,t)$ is either $p(y,t)- v_n$ or $p(y+1, t) - v_n$, which are both bounded from below by $v_1-v_n$ and from above by $v_N-v_n$, hence $M$ is an $\F^I$-martingale (cf. \cite[Chapter I, T6]{Bremaud}). On the right hand side of \eqref{eq: U ito}, splitting the second integral on $\{Y_{r-} \geq \um\}$, $\{Y_{r-}=\dm\}$, and $\{Y_{r-} < \dm\}$, splitting the fourth integral on $\{Y_{r-} > \um\}$, $\{Y_{r-} = \um\}$, and $\{Y_{r-} \leq \dm\}$, utilizing $U(\um, \cdot) = U(\dm, \cdot)$, as well as different equations in Lemma \ref{lem: U eqn} in different regions, we obtain
 \begin{align*}
  &U(Y_1, 1)\\
  &= U(y,t) + \int_t^1 \pare{p(\dm, r) - v_n} \beta \indic_{\{Y_{r-} = \um\}} dr + \int_t^1 \pare{v_n - p(\um, r)} \beta \indic_{\{Y_{r-} = \dm\}} dr\\
  & \quad - \int_t^1 \pare{v_n - p(Y_{r-}, r)} \indic_{\{Y_{r-} \geq \um\}} (\theta^{B,B}_r- \theta^{S,B}_r) dr - \int_t^1 \pare{v_n - p(Y_{r-}+1, r)} \indic_{\{Y_{r-} < \dm\}} (\theta^{B,B}_r - \theta^{S,B}_r) dr\\
  & \quad - \int_t^1 \pare{v_n - p(Y_{r-}+1, r)} \indic_{\{Y_{r-}\geq \dm\}} \theta^{B,T}_r dr - \int_t^1 \pare{v_n - p(Y_{r-}+2, r)} \indic_{\{Y_{r-}< \dm-1\}} \theta^{B,T}_r dr\\
  & \quad + \int_t^1 \pare{v_n - p(Y_{r-}-1, r)} \indic_{\{Y_{r-} > \um\}} (\theta^{S,S}_r - \theta^{B,S}_r) dr + \int_t^1 \pare{v_n - p(Y_{r-}, r)} \indic_{\{Y_{r-} \leq \dm\}} (\theta^{S,S}_r - \theta^{B,S}_r) dr\\
  & \quad + \int_t^1 \pare{v_n - p(Y_{r-} -2, r)} \indic_{\{Y_{r-} > \um + 1\}} \theta^{S,T}_r dr + \int_t^1 \pare{v_n - p(Y_{r-} -1,r)} \indic_{\{Y_{r-}\leq \um\}} \theta^{S,T}_r dr\\
  & \quad + M_1 - M_t.
 \end{align*}
 Rearranging the previous identity by putting the profit of $(X^B, X^S)$ to the left hand side, we obtain
 \begin{equation}\label{eq: profit iden}
 \begin{split}
  &\int_t^1 (v_n - p(Y_{r-}+1, r)) \theta^{B,B}_r dr + \int_t^1 \pare{v_n - p(Y_{r-}+2, r)}\theta^{B,T}_r dr+ \int_t^1 (v_n - p(Y_{r-}, r)) \theta^{B,S}_r dr\\
  &- \int_t^1 (v_n - p(Y_{r-}-1, r)) \theta^{S,S}_r dr- \int_t^1 \pare{v_n - p(Y_{r-}-2, r)} \theta^{S,T}_r dr - \int_t^1 (v_n - p(Y_{r-}, r)) \theta^{S,B}_r dr\\
  =& U(y,t) - U(Y_1, 1) - K - L + M_1 - M_t,
 \end{split}
 \end{equation}
 where
 \begin{align*}
  K= &\quad \int_t^1 \pare{p(Y_{r-}+1, r) - p(Y_{r-}, r)} \indic_{\{Y_{r-} \geq \um\}} \theta^{B,B}_r dr + \int_t^1 \pare{p(Y_{r-}, r) -p(Y_{r-}-1, r)} \indic_{\{Y_{r-} \geq \um\}} \theta^{B,S}_r dr\\
  & \quad \int_t^1 \pare{p(Y_{r-}+2, r)-p(Y_{r-} +1, r)} \indic_{\{Y_{r-}\geq \dm\}} \theta^{B,T}_r dr\\
  & + \int_t^1 \pare{p(Y_{r-}, r) - p(Y_{r-}-1, r)} \indic_{\{Y_{r-}\leq \dm\}} \theta^{S,S}_r dr + \int_t^1 \pare{p(Y_{r-}+1, u) - p(Y_{r-}, r)} \indic_{\{Y_{r-} \leq \dm\}} \theta^{S,B}_r dr\\
  & \quad + \int_t^1 \pare{p(Y_{r-}-1, r) -p(Y_{r-}-2, r)} \indic_{\{Y_{r-} \leq \um\}}\theta^{S,T}_r dr,\\
  L= & \quad \int_t^1 [v_n - p(\dm, r)] \bra{(\beta- \theta^{B,S}_r + \theta^{S,S}_r) \indic_{\{Y_{r-}= \um\}} + \theta^{S,T}_r \indic_{\{Y_{r-}=\um +1\}}}dr\\
  &  - \int_t^1 [v_n- p(\um, r)]\bra{(\beta- \theta^{S,B}_r + \theta^{B,B}_r) \indic_{\{Y_{r-}= \dm\}} + \theta^{B,T}_r \indic_{\{Y_{r-}= \dm-1\}}}dr.
 \end{align*}
 Taking conditional expectation $\expec[\cdot |\F^I_t, Y_t =y]$ on both sides of \eqref{eq: profit iden}, the left hand side is the expected profit $\mathcal{J}(X^B, X^S)$, while, on the right hand side, both $U(\cdot, 1)$ and $K$ are nonnegative (cf. Definition \ref{def: price rule} i)). Therefore \eqref{eq: ep<U-L} is verified. To attain the identity in \eqref{eq: ep<U-L}, we need i) $Y_1 \in [a_n-1, a_{n+1}+1)$ a.s. so that $U(Y_1, 1) =0$ a.s. follows from \eqref{eq: U=0}; ii) $\theta^{B,B}= \theta^{B,S} \equiv 0$ when $y\geq \um$, $\theta^{S,S}= \theta^{S,B}\equiv 0$ when $y\leq \dm$, $\theta^{B,T} \equiv 0$ when $y\geq \dm$, and $\theta^{S,T}\equiv 0$ when $y \leq \um$.
\end{proof}

Come back to the statement of Proposition \ref{prop: ep<U-L}. If the insider chooses a strategy such that both conditions in i) and ii) are satisfied, then the identity in \eqref{eq: ep<U-L} is attained, hence the expected profit of this strategy is $U-L$. On the other hand, define $U^S: \{v_1, \cdots, v_N\} \times \Integer \times [0,1] \rightarrow \Real$ via
\begin{equation}\label{eq: oU}
 U^S(v_n,y, t) = \left\{\begin{array}{ll}U(v_n, y, t) & y\geq \um\\ U(v_n, y-1, t) & y\leq \dm\end{array}\right..
\end{equation}
The next result shows that $U^S$  dominates the value function $V$, therefore $U^S-U+L$ is the upper bound of the potential loss of the expected profit. In Section \ref{sec: convergence}, we will prove this potential loss converges to zero as $\delta\downarrow 0$. Therefore, when the order size is small, the insider losses little expected profit by employing a strategy satisfying Proposition \ref{prop: ep<U-L} i) and ii).

\begin{prop}\label{prop: U>V}
 Let Assumption \ref{ass: middle level} hold. Then $V\leq U^S$, hence $V<\infty$, on $\{v_1, \cdots, v_N\}\times \Integer \times [0,1]$.
\end{prop}

\begin{proof}
 Fix $v_n$ and omit it as the first argument of $U^S$ and $U$ throughout the proof. We first verify
\begin{align}
 &U^S(y,t) - U^S(y+1, t) - (v_n - p(y,t)) =0,\label{eq: US 1}\\
 &U^S_t + \pare{U^S(y+1, t) - 2U^S(y,t) + U^S(y-1,t)} \beta=0 \label{eq: US 2},
\end{align}
for any $(y,t)\in \Integer \times [0,1)$. Indeed, when $y\geq \um$, \eqref{eq: US 1} is exactly \eqref{eq: U eqn >um}. When $y=\dm$,
\[
 U^S(\dm, t) - U^S(\um, t) = U(\dm-1, t) - U(\um, t) = U(\dm-1, t) - U(\dm, t) = v_n - p(\dm, t),
\]
where the second identity follows from $U(\um, t) = U(\dm, t)$ and the third identity holds due to \eqref{eq: U eq <dm}. When $y<\dm$,
\[
 U^S(y, t) - U^S(y+1, t) = U(y-1, t) - U(y,t) = v_n - p(y,t),
\]
where \eqref{eq: U eq <dm} is utilized again to obtain the second identity. Therefore \eqref{eq: US 1} is confirmed for all cases. As for \eqref{eq: US 2}, \eqref{eq: US 1} yields
\[
 U^S(y+1, t) - 2U^S(y,t) + U^S(y-1, t) = p(y,t) - p(y-1, t).
\]
On the other hand, we have from \eqref{eq: U um} and \eqref{eq: U dm} that
\[
 U^S_t(y,t)= \left\{\begin{array}{ll}U_t(y,t) = -\beta (p(y,t)-p(y-1,t)) & y\geq \um\\ U_t(y-1, t) = -\beta(p(y,t)-p(y-1, t)) & y\leq \dm\end{array}\right..
\]
Therefore \eqref{eq: US 2} is confirmed after combining the previous two identities.

Now note that $U^S(\cdot, 1) \geq 0$, moreover $U^S$ satisfies \eqref{eq: US 1} and \eqref{eq: US 2}. The assertion $V\leq U^S$ follows from the same argument as in the high type of \cite[Proposition 3.2]{Cetin-Xing}.
\end{proof}

Having studied the insider's optimization problem, let us turn to the market maker. Given $(X^B, X^S; \F^I)$, Definition \ref{def: Asy GM eq} ii) requires the pricing rule to be rational. This leads to another constraint on $(X^B, X^S; \F^I)$.

\begin{prop}\label{prop: rational}
If there exists an admissible strategy $(X^B, X^S; \F^I)$ such that
 \begin{enumerate}[i)]
  \item $Y^B= Z^B + X^{B,B}+ X^{B,T}- X^{S,B}$ and $Y^S= Z^S + X^{S,S} + X^{S,T} - X^{B,S}$ are independent $\F^Y-$adapted Poisson processes with common intensity $\beta$;
  \item $[Y_1 \in [a_n, a_{n+1})] = [\tv=v_n]$, $n=1, \cdots, N$.
 \end{enumerate}
 Then the pricing rule \eqref{eq: p func} is rational.
\end{prop}

\begin{proof}
 For any $t\in [0,1]$,
 \begin{align*}
  p(Y_t, t) &= \expec^{Y_t}[P(Z_{1-t})] = \expec\bra{P(Z_1)\,|\, Z_t = Y_t}= \expec\bra{P(Y_1)\,|\, \F^Y_t} = \expec[\tv\,|\, \F^Y_t],
 \end{align*}
 where the third identity holds since $Y$ and $Z$ have the same distribution, the fourth identity follows from ii) and \eqref{eq: P}.
\end{proof}

\begin{rem}\label{rem: no top up}
 If the insider places a buy (resp. sell) order when a noise buy (resp. sell) order arrives, Proposition \ref{prop: rational} i) cannot be satisfied. Therefore in the asymptotic equilibrium the insider will not trade in the same direction as the noise traders, i.e., $X^{B,T}=X^{S,T}\equiv 0$, so that the market maker can employ a rational pricing rule.
\end{rem}

Concluding this section, we need to construct point processes $(X^B, X^S; \F^I)$ which simultaneously satisfy conditions in Proposition \ref{prop: ep<U-L} ii), Proposition \ref{prop: rational} i) and ii)\footnote{Note is Proposition \ref{prop: rational} ii) implies Proposition \ref{prop: ep<U-L} i).}. This construction is a natural extension of \cite[Section 4]{Cetin-Xing}, where $N=2$ is considered, and will be presented in the next section.

\section{Construction of a point process bridge}\label{sec: point process bridge}

In this section, we will construct point processes $X^B$ and $X^S$ on a probability space $(\Omega, \F^I, (\F^I_t)_{t\in [0,1]}, \prob)$ such that $X^{B,T}= X^{S,T}\equiv 0$, due to Remark \ref{rem: no top up}, and satisfy
\begin{enumerate}[i)]
 \item $Y^B = Z^B + X^{B,B} - X^{S,B}$ and $Y^S = Z^S + X^{S,S} - X^{B,S}$ are independent $\F^Y$-adapted Poisson processes with common intensity $\beta$;
 \item $X^{B,B}_t = X^{B,S}_t \equiv 0$ when $Y_{t-}\geq \um$, $X^{S,S}_t = X^{S,B}_t \equiv 0$ when $Y_{t-}\leq \dm$;
 \item $[Y_1 \in [a_n, a_{n+1})]= [\tv=v_n]$ $\prob$-a.s. for $n=1, \cdots, N$.
\end{enumerate}
The construction is a natural extension of \cite{Cetin-Xing} where $N=2$ is considered. As in \cite{Cetin-Xing}, $X^B$ and $X^S$ are constructed using two independent sequences of iid random variables $(\eta_i)_{i\geq 1}$ and $(\zeta_i)_{i\geq 1}$ with uniform distribution on $[0,1]$, moreover they are independent of $Z$ and $\tv$. The insider uses $(\eta_i)_{i\geq 1}$ to randomly contribute either buy or sell orders, and uses $(\zeta_i)_{i\geq 1}$ to randomly cancel noise orders. Throughout this section Assumption \ref{ass: middle level} is enforced. Moreover, we set $\delta=1$, hence suppress the superscript $\delta$. Otherwise  $X^B$ and $X^S$ can be scaled by $\delta$ to obtain the desired processes.

In the following construction, we will define a probability space $(\Omega, \F^I, (\F^I_t)_{t\in [0,1]}, \prob)$ on which $Y$ takes the form
\begin{equation}\label{eq: Y decomp}
 Y= Z+ \sum_{n=1}^N \indic_{A_n}(X^B-X^S).
\end{equation}
Here $Z$ is the difference of two independent $\F^I$-adapted Poisson processes with intensity $\beta$, $A_n\in \F^I_0$  such that $\prob(A_n) = \prob(Z_1\in [a_n, a_{n+1}))$ for each $n=1, \cdots, N$.

Before constructing $X^B$ and $X^S$ satisfying desired properties, let us draw some intuition from the theory of filtration enlargement. Let us define $(\mathbb{D}([0,1], \Integer), \overline{\F}, (\overline{\F}_t)_{t\in [0,1]}, \overline{\prob})$ be the canonical space where $\mathbb{D}([0,1], \Integer)$ is $\Integer$-valued \cadlag\, functions, $\overline{\prob}$ is a probability measure under which $Z^B$ and $Z^S$ are independent Poisson processes with intensities $\beta$, $(\overline{\F}_t)_{t\in [0,1]}$ is the minimal filtration generated by $Z^B$ and $Z^S$ satisfying the usual conditions, and $\overline{\F} = \vee_{t\in [0,1]} \overline{\F}_t$. Let us denote by $(\mathcal{G}_t)_{t\in [0,1]}$ the filtration $(\overline{\F}_t)_{t\in [0,1]}$ enlarged with a sequence of random variables $(\indic_{\{Z_1 \in [a_n, a_{n+1})\}})_{n=1, \cdots, N}$.

In order to find the $\mathcal{G}$-intensities of $Z^B$ and $Z^S$, we use a standard enlargement of filtration argument which can be found, e.g. in \cite{Mansuy-Yor}. To this end, recall $h_n(y,t) = \overline{\prob}[Z_1 \in [a_n, a_{n+1}) \,|\, Z_t = y]$. Note that $h_n$ is strictly positive on $\Integer \times [0,1)$. Moreover the Markov property of $Z$ implies $h_n$ is continuously differentiable in the time variable and satisfies
\begin{equation}\label{eq: hn pde}
\begin{split}
 &\partial_t h_n + \pare{h_n(y+1, t) - 2h_n(y,t) + h_n(y-1, t)}\beta =0, \quad (y,t)\in \Integer \times [0,1),\\
 & h_n(y,1) = \indic_{\{y\in [a_n, a_{n+1})\}}.
\end{split}
\end{equation}

\begin{lem}\label{lem: Zintensity}
The $\mathcal{G}$-intensities of $Z^{B}$ and $Z^{S}$ at $t \in [0,1)$ are given by
\[
\sum^{N}_{n=1}\indic_{\{Z_{1} \in [a_{n}, a_{n+1})\}}\frac{h_{n}(Z_{t-}+1,t)}{h_{n}(Z_{t-},t)}\beta \text{\quad and \quad } \sum^{N}_{n=1}\indic_{\{Z_{1} \in [a_{n}, a_{n+1})\}}\frac{h_{n}(Z_{t-}-1,t)}{h_{n}(Z_{t-},t)}\beta,
\]
respectively.
\end{lem}

\begin{proof}
We will only calculate the intensity for $Z^{B}$. The intensity of $Z^{S}$ can be obtained similarly. All expectations are taken under $\overline{\prob}$ throughout this proof. For $s\leq t <1$, take an arbitrary $E \in \overline{\F}_s$ and denote $M^{B}_t := Z_t^{B} - \beta t$. The definition of $h_n$ and the $\overline{\F}$-martingale property of $M^{B}$ imply
\begin{align*}
&\expec\bra{(M_t^{B}-M_s^{B})\indic_E \indic_{\{Z_1 \in [a_n,a_{n+1})\}}} \\
& = \expec\bra{(M_t^{B}-M_s^{B})\indic_E h_n(Z_t,t)} \\
& = \expec\bra{\indic_E(\langle M^{B}, h_n(Z_\cdot,\cdot)\rangle_t - \langle M^{B}, h_n(Z_\cdot,\cdot)\rangle_s)} \\
& = \expec\bra{\indic_E\int^t_s \beta \pare{h_n(Z_{r-}+1,r)-h_n(Z_{r-},r)}dr  } \\
& = \expec\bra{\indic_E\int^t_s\beta\,\indic_{\{Z_1 \in [a_{n}, a_{n+1})\}}\frac{h_n(Z_{r-}+1,r)-h_n(Z_{r-},r)}{h_n(Z_{r-},r)}dr}.
\end{align*}
These computations for each $n=1, \cdots, N$ imply that
\[
	M^{B}-\int^\cdot_s\beta\sum^N_{n=1}\indic_{\{Z_1 \in [a_{n}, a_{n+1})\}}\frac{h_n(Z_{r-}+1,r)-h_n(Z_{r-},r)}{h_n(Z_{r-},r)}dr
\]
defines a $\mathcal{G}$-martingale. Therefore the $\mathcal{G}$-intensity of $Z^{B}$ follows from $Z_t^{B}=M_t^{B}+\beta t$.
\end{proof}

To better understand intensities in the previous lemma, let us collect several properties for $h_n$:

\begin{lem}\label{lem: hn property}
 Let Assumption \ref{ass: middle level} hold. The following properties hold for each $h_n$, $n=1, \cdots, N$:
 \begin{enumerate}[i)]
  \item $h_n(\cdot,\cdot) = h_n(2m_n-\cdot, \cdot)$; in particular, $h_n(\um, \cdot) = h_n(\dm, \cdot)$.
  \item $y\mapsto h_n(y,t)$ is strictly increasing when $y\leq \dm$ and strictly decreasing when $y\geq \um$.
 \end{enumerate}
 Here, when $n=1$ (resp. $n=N$), $\um=\dm=-\infty$ (resp. $\um=\dm= \infty$).
\end{lem}

\begin{proof}
 Recall that $a_n+a_{n+1}-1=2m_n$. Then
 \begin{align*}
  h_n(y,t)&= \overline{\prob}[Z_1\in [a_n, a_{n+1}) \,|\, Z_t = y] = \overline{\prob}[y+ Z_{1-t} \in [a_n, a_{n+1})] \\
  &= \overline{\prob}[2m_n -y- Z_{1-t} \in (2m_n -a_{n+1}, 2m_n - a_n]]\\
  &= \overline{\prob}[2m_n-y-Z_{1-t} \in [a_n, a_{n+1})]=h_n(2m_n-y, t),
 \end{align*}
 where the last identity holds since $Z$ and $-Z$ have the same distribution. This verifies i). To prove ii), rewrite $h_n(y,t)= \overline{\prob}[Z_{1-t}\in [a_n-y, a_{n+1}-y)]$. Then the statement ii) follows from the fact that $y\mapsto \prob(Z_{1-t}=y)$ is strictly increasing when $y\leq 0$ and strictly decreasing when $y\geq 0$.
\end{proof}

In what follows, given $A_n \in \F^I_0$ such that $\prob (A_n) = \overline{\prob}(Z_1 \in [a_n, a_{n+1}))$, $(X^B, X^S; \F^I)$ on $A_n$ will be constructed so that $\F^I$-intensity of $Y^B$ (resp. $Y^S$) on $A_n$ match $\mathcal{G}$-intensities of $Z^B$ (resp. $Z^S$) on $[Z_1 \in [a_n, a_{n+1})]$. Matching these intensities ensures that $(X^B, X^S; \F^I)$ satisfies desired properties, cf. Proposition \ref{prop: Y property} below. Recall $Y^B= Z^B + X^{B,B}- X^{S, B}$ and $Y^S=Z^S + X^{S,S} - X^{B,S}$. Subtracting $\beta$ from $\mathcal{G}$-intensities of $Z^B$ (resp. $Z^S$) in Lemma \ref{lem: Zintensity}, we can read out intensities of $X^{B,B}- X^{S,B}$ (resp. $X^{S,S}- X^{B,S}$). Since property ii) at the beginning of this section implies that $\theta^B$ and $\theta^S$ are never positive at the same time. Therefore, when the intensity of $X^{B,B}-X^{S,B}$ is positive, the insider contributes buy orders $X^{B,B}$ with such intensity, otherwise the insider submits sell orders $X^{S,B}$ with the same intensity to cancel some noise buy orders from $Z^B$. Applying the same strategy to $X^{S,S}- X^{B,S}$ and utilizing Lemma \ref{lem: hn property}, we read out $\F^I$-intensities for $X^{i,j}$, $i,j\in \{B,S\}$:

\begin{cor}\label{Cor: X intensity}
 Suppose that $\F^I$-intensities of $Y^B$ and $Y^S$ match $\mathcal{G}$-intensities of $Z^B$ and $Z^S$ respectively, moreover $X^{B,B}_t = X^{B,S}_t \equiv 0$ when $Y_{t-} \geq \um$ and $X^{S,S}_t = X^{S,B}_t\equiv 0$ when $Y_{t-}\leq \dm$. Then $\F^I$-intensities of $X^{i,j}$, $i,j\in \{B,S\}$, have the following form on $A_n$ when $Y_{t-}=y$:
 \begin{align*}
  &\theta^{B,B}(y,t) = \pare{\frac{h_n(y+1, t)}{h_n(y,t)}-1}_+ \beta, &\theta^{B,S}(y,t) = \pare{\frac{h_n(y-1, t)}{h_n(y,t)}-1}_- \beta,\\
  &\theta^{S,S}(y,t) = \pare{\frac{h_n(y-1, t)}{h_n(y,t)} -1}_+ \beta,
  &\theta^{S,B}(y,t) = \pare{\frac{h_n(y+1, t)}{h_n(y,t)} -1}_- \beta.
 \end{align*}
 In particular, $\theta^{i,j}$, $i,j\in \{B,S\}$, satisfies the following properties:
 \begin{enumerate}[i)]
  \item $\theta^{B,B}(y,\cdot) = \theta^{B,S}(y,\cdot) \equiv 0$, $\theta^{S,S}(y,\cdot)>0$, and $\theta^{S,B}(y,\cdot)>0$, when $y\geq \um$; $\theta^{S,S}(y,\cdot) =\theta^{S,B}(y,\cdot) \equiv 0$, $\theta^{B,B}(y,\cdot)>0$, and $\theta^{B,S}(y,\cdot)>0$, when $y\leq \dm$;
  \item $\theta^{B,B}(\cdot, \cdot) = \theta^{S,S}(2m_n - \cdot, \cdot)$, $\theta^{B,S}(\cdot, \cdot) = \theta^{S,B}(2m_n-\cdot, \cdot)$;
  \item $\theta^{B,B}(\dm, \cdot) = \theta^{S,S}(\um, \cdot)\equiv 0$.
 \end{enumerate}
\end{cor}

As  described in Corollary \ref{Cor: X intensity}, when $A_n \in \overline{\F}_0$ is fixed, the state space is divided into two domains $\mathcal{S}:= \{y\in \Integer: y\geq \um\}$ and $\mathcal{B}:= \{y\in \Integer: y\leq \dm\}$. As $Y$ making excursions into these two domains, either $X^S$ or $X^B$ is active. In the following construction, we will focus on the domain $\mathcal{B}$ and construct inductively jumps of $X^B$ until $Y$ leaves $\mathcal{B}$. When $Y$ excurses in $\mathcal{S}$, $X^S$ can be constructed similarly.

When $Y$ is in $\mathcal{B}$, one of the goals of $X^B$ is to make sure that $Y_1$ ends up in the interval $[a_n, a_{n+1})$. In order to achieve this goal, $X^B$ will add some jumps in addition to the jumps coming from $Z^B$. However this by itself will not be enough since $Y$ also jumps downward due to $Z^S$. Thus, $X^B$ also needs to cancel some of downwards jumps from $Z^S$. Therefore $X^B$ consists of two components $X^{B,B}$ and $X^{B,S}$, where $X^{B,B}$ complements jumps of $Z^B$ and $X^{B,S}$ cancels some jumps of $Z^S$. Let us denote by $(\tau_i)_{i\geq 1}$ the sequence of jump times for $Y$. These stopping times will be constructed inductively as follows. Given $\tau_{i-1}<1$ and $Y_{\tau_{i-1}}\leq \dm$, the next jump time $\tau_i$ happens at the minimum of the following three random times:
\begin{itemize}
 \item the next jump of $Z^B$,
 \item the next jump of $X^{B,B}$,
 \item the next jump of $Z^S$ which is not canceled by a jump of $X^{B,S}$.
\end{itemize}
Here $X^{B,B}$ and $X^{B,S}$ need to be constructed so that their intensities $\theta^{B,B}(Y_{t-}, t)$ and $\theta^{B,S}(Y_{t-}, t)$ match the forms in Corollary \ref{Cor: X intensity}. This goal is achieved by employing two independent sequences of iid random variables $(\eta_i)_{i\geq 1}$ and $(\zeta_i)_{i\geq 1}$ with uniform distribution on $[0,1]$. They are also independent of $\overline{\mathcal{F}}$ and $(A_n)_{n=1, \cdots, N}$. These two sequences will be used to generate a random variable $\nu_i$ and another sequence of Bernoulli random variables $(\xi_{j,i})_{j\geq 1}$ taking values in $\{0,1\}$. Let $(\sigma_i^+)_{i\geq 1}$ and $(\sigma^-_i)_{i\geq 1}$ be jump time of $Z^B$ and $Z^S$, respectively. Then, after $\tau_{i-1}$, the next jump of $Z^B$ is at $\sigma^+_{Z_{\tau_{i-1}}^B +1}$, the next jump of $X^{B,B}$ is at $\nu_i$, and the next jump of $Z^S$ not canceled by jumps of $X^{B,S}$ is at $\tau^-_i= \min\{\sigma^-_j >\tau_{i-1}\,:\, \xi_{j,i}=1\}$. Then the next jump of $Y$ is at
\[
 \tau_i = \sigma^+_{Z^B_{\tau_{i-1}} +1} \wedge \nu_i \wedge \tau^-_i.
\]
The construction of $\nu_i$ and $(\xi_{j,i})_{j\geq 1}$ using $(\eta_i)_{i\geq 1}$ and $(\zeta_i)_{i\geq 1}$ is exactly the same as in \cite[Section 4]{Cetin-Xing}, only replacing $h$ therein by $h_n$.

All aforementioned construction is performed in a filtrated probability space $(\Omega, \F^I, (\F^I_t)_{t\in [0,1]}, \prob)$ such that there exist $(A_n)_{n=1, \cdots, N}\in \F^I_0$ with $\prob(A_n) = h_n(0,0)$ and two independent sequences of iid $\F^I$-measurable random variables $(\eta_i)_{i\geq 1}$ and $(\zeta_i)_{i\geq 1}$ with uniform distribution on $[0,1]$, moreover these two sequences are independent of both $Z$ and $(A_n)_{n=1, \cdots, N}$. These requirements can be satisfied by extending $\overline{\F}_0$ (resp. $\overline{\F}$) to $\F^I_0$ (resp. $\F^I$). As for the filtration $(\F^I_t)_{t\in[0,1]}$, we require that it is right continuous and complete under $\prob$, moreover $Z$, as the difference of two independent Poisson processes with intensity $\beta$, is adapted to $(\F^I_t)_{t\in [0,1]}$. Therefore $Z$ is independent of $(A_n)_{n=1, \cdots, N}$, since $Z$ has independent increments. Finally, we also assume that $(\F^I_t)_{t\in[0,1]}$ is rich enough so that $(\nu_i)_{i\geq 1}$ and $(\tau^-_i)_{i\geq 1}$ discussed above are $\F^I$-stopping times.

An argument similar to \cite[Lemma 4.3]{Cetin-Xing} yields:
\begin{lem}\label{lem: Y intensity}
 Given point processes $(X^B, X^S; \F^I)$ constructed above, the $\F^I$-intensities of $Y^B$ and $Y^S$ at $t\in[0,1)$ are given by
 \[
  \sum_{n=1}^N \indic_{A_n} \frac{h_n(Y_{t-}+1, t)}{h_n(Y_{t-}, t)} \beta \quad \text{ and } \quad \sum_{n=1}^N \indic_{A_n} \frac{h_n(Y_{t-}-1, t)}{h_n(Y_{t-}, t)} \beta,
 \]
 respectively.
\end{lem}

Now we are ready to verify that our construction is as desired.

\begin{prop}\label{prop: Y property}
The process $Y$ as constructed above satisfies the following properties:
\begin{enumerate}[i)]
\item $[Y_1 \in [a_n,a_{n+1})] = A_n$ a.s. for $n=1,\cdots,N$;
\item $Y^B$ and $Y^S$ are independent Poisson processes with intensity $\beta$ with respect to the natural filtration $(\F_t^Y)_{t\in[0,1]}$ of $Y$;
\item $(X^B, X^S; \F^I)$ is admissible in the sense of Definition \ref{def: insider ad}.
\end{enumerate}
\end{prop}
\begin{proof}
 To verify that $Y$ satisfies the desired properties, let us introduce an auxiliary process $(\ell_t)_{t\in[0,1)}$:
\[
	\ell_t:=\sum^N_{n=1}\indic_{A_n}\frac{h_n(0,0)}{h_n(Y_t,t)} \quad t \in [0,1).
\]
When $n=2, \cdots, N-1$, there is only almost surely finite number of positive (resp. negative) jumps of $Y$ on $A_n$ when $Y_\cdot \geq \um$ (resp. $Y_\cdot \leq \dm$). Therefore $Y_t$ is finite on these $A_n$ when $t<1$ is fixed. When $n=1$ (resp. $n=N$), there is finite number of positive (resp. negative) jumps of $Y$ on $A_1$ (resp. $A_N$) before $t$. Hence $Y_t<\infty$ on $A_1$ (resp. $Y_t>-\infty$ on $A_N$). This analysis implies $h_n(Y_t, t)>0$ on $A_n$ for each $n=1, \cdots, N$ and $t<1$. Therefore $(\ell_t)_{t\in [0,1)}$ is well defined positive process with $\ell_0=1$.

To prove i), we first show that $\ell$ is a positive $\F^I$-local martingale on $[0,1)$. To this end, It\^o formula yields that
\[
	d\ell_t = \sum^N_{n=1}\indic_{A_n}\ell_{t-}\bra{\frac{h_n(Y_{t-},t)-h_n(Y_{t-}+1,t)}{h_n(Y_{t-}+1,t)}dM_t^B + \frac{h_n(Y_{t-},t)-h_n(Y_{t-}-1,t)}{h_n(Y_{t-}-1,t)}dM_t^S}, \quad t \in [0,1).
\]
Here
\[
M^B = Y^B- \beta\int^\cdot_0\sum^N_{n=1}\indic_{A_n}\frac{h_n(Y_{r-}+1,r)}{h_n(Y_{r-},r)}dr, \quad M^S = Y^S- \beta\int^\cdot_0\sum^N_{n=1}\indic_{A_n}\frac{h_n(Y_{r-}-1,r)}{h_n(Y_{r-},r)}dr
\]
are all $\F^I$-local martingales. Define $\zeta^+_m=\inf\{t \in [0,1]:Y_t=m\}$ and $\zeta^-_m=\inf\{t \in [0,1]:Y_t=-m\}$. Consider the sequence of stopping time $(\eta_m)_{m \geq 1}$:
\[
	\eta_m:=\pare{\indic_{\cup_{n=2}^{N-1} A_n} \zeta^+_m \wedge \zeta^-_m + \indic_{A_1} \zeta^+_m + \indic_{A_N} \zeta^-_m}\wedge (1-1/m).
\]
It follows from the definition of $h_n$ that each $h_n(Y_t, t)$ on $A_n$ is bounded away from zero uniformly in $t\in [0, \eta_m]$. This implies that $\ell^{\eta_m}$ is bounded, hence $\ell^{\eta_m}$ is an $\F^I$-martingale. The construction of  $Y$ yields $\lim_{m \to \infty} \eta_m = 1$. Therefore, $\ell$ is a positive $\F^I$-local martingale, hence also a supermartingale, on $[0,1)$.

Define $\ell_1 := \lim_{t \to 1} \ell_t$, which exists and is finite due to Doob's supermartingale convergence theorem. This implies $h_n(Y_{1-},1) > 0$ on $A_n$. On the other hand, the construction of $Y$ yields $Y^S$ (resp. $Y^B$) does not jump at time $1$ $\prob$-a.s. when $Y_{1-}\leq \dm$ (resp. $Y_{1-}\geq \um$). Therefore $h_n(Y_1, 1)>0$ on $A_n$. However $h_n(\cdot, 1)$ by definition can only be either $0$ or $1$. Hence $Y_1\in [a_n, a_{n+1})$ on $A_n$, for each $n=1, \cdots, N$, and the statement i) is confirmed.

As for the statement ii), we will prove that $Y^B$ is an $\F^Y$-adapted Poisson process. The similar argument can be applied to $Y^S$ as well. In view of the $\F^I$-intensity of $Y^B$ calculated in Lemma \ref{lem: Y intensity}, one has that, for each $i \geq 1$,
\[
Y^B_{\cdot \wedge \tau_i\wedge 1} - \beta\pare{\int^{\cdot\wedge \tau_i \wedge 1}_0 \sum^N_{n=1}\indic_{A_n} \frac{h_n(Y_{u-}+1,u)}{h_n(Y_{u-},u)}du }
\]
is an $\F^I$-martingale, where $\tau_i$ is the $i^{th}$ jump time of $Y$. We will show in the next paragraph that, when stopped at $\tau_i \wedge 1$, $Y^B$ is Poisson process in $\F^Y$ by showing that $(Y^B_{\tau_i \wedge t}-\beta(\tau_i \wedge t))_{t \in [0,1]}$ is an $\F^Y$-martingale. (Here note that $\tau_i$ is an $\F^Y$-stopping time.) This in turn will imply that $Y^B$ is a Poisson process with intensity $\beta$ on $[0,\tau \wedge 1)$ where $\tau=\lim_{i\to\infty}\tau_i$ is the explosion time. Since Poisson process does not explode, this will further imply $Y^B_{\tau\wedge 1} < \infty$ and, therefore, $\tau \geq 1$, $\prob$-a.s..

We proceed by projecting the above martingale into $\F^Y$ to see that
\[
	Y^B - \beta\int^\cdot_0\sum^N_{n=1}\prob(A_n|\F^Y_r)\frac{h_n(Y_{r-}+1,r)}{h_n(Y_{r-},r)}dr
\]
is an $\F^Y$-martingale when stopped at $\tau_i \wedge 1$. Therefore, it remains to show that, for almost all $t \in [0,1)$, on $[t \leq \tau_i]$,
\begin{equation}\label{eq: h sum id}
\sum^N_{n=1}\prob(A_n|\F^Y_t)\frac{h_n(Y_{t-}+1,t)}{h_n(Y_{t-},t)} = 1, \quad \prob\text{-a.s.}.
\end{equation}
To this end, we will show, on $[t \leq \tau_i]$,
\begin{equation}
\label{pIn}
	\prob(A_n|\F^Y_t) = h_n(Y_{t},t), \quad \text{for } t \in [0,1).
\end{equation}
Then \eqref{eq: h sum id} follows since $Y_t \neq Y_{t-}$ only for countably many times.

We have seen that $(\ell_{u\wedge\tau_i})_{u\in[0,t]}$ is a strictly positive $\F^I$-martingale for each $i$. Define a probability measure $\qprob^i \sim \prob$ on $\F^I_t$ via $d\qprob^i/d\prob|_{\F^I_t} = \ell_{\tau_i\wedge t}$. It follows from Girsanov's theorem that $Y^B$ is a Poisson process when stopped at $\tau_i \wedge t$ and with intensity $\beta$ under $\qprob^i$. Therefore, they are independent from $A_n$ under $\qprob^i$. Then, for $t < 1$, we obtain from the Bayes's formula that
\begin{equation}
\label{In}
\begin{split}
\indic_{\{r\leq\tau_i\wedge t\}}\prob(A_n|\F^Y_r) &= \indic_{\{r \leq \tau_i \wedge t\}}\frac{\expec^{\qprob^i}[\indic_{A_n}\ell_r^{-1}|\F^Y_r]}{\expec^{\qprob^i}[\ell_r^{-1}|\F^Y_r]} \\
&=\indic_{\{r \leq \tau_i \wedge t\}}\frac{\expec^{\qprob^i}[\indic_{A_n}\frac{h_n(Y_r,r)}{h_n(0,0)}|\F^Y_r]}{\expec^{\qprob^i}[\sum^N_{n=1}\indic_{A_n}\frac{h_n(Y_r,r)}{h_n(0,0)}|\F^Y_r]} \\
& = \indic_{\{r \leq \tau_i \wedge t\}}h_n(Y_r,r),
\end{split}
\end{equation}
where the third identity follows from the aforementioned independence of $Y$ and $A_n$ under $\qprob^i$ along with the fact that $\qprob^i$ does not change the probability of $\F_0^I$ measurable events so that $\qprob^i(A_n) = \prob(A_n) = h_n(0,0)$. As result, (\ref{pIn}) follows from (\ref{In}) after sending $i \to \infty$.

Since $Y^B$ and $Y^S$ are $\F^Y$-Poisson processes and they do not jump simultaneously by their construction, they are then independent. To show the strategy $(X^B, X^S; \F^I)$ constructed is admissible, it remains to show both $\expec[X^B_1 \indic_{A_n}]$ and $\expec[X^S_1 \indic_{A_n}]$ are finite for each $n=1, \cdots, N$. To this end, for each $n$, $\expec[X^B_1 \indic_{A_n}]= \expec[X^{B,B}_1 \indic_{A_n}] + \expec[X^{B,S}_1 \indic_{A_n}]$, where $\expec[X^{B,S}_1 \indic_{A_n}] \leq \expec[Z^S]<\infty$ and $\expec[X^{B,B}_1 \indic_{A_n}] \leq \expec[Y^B_1 \indic_{A_n}] + \expec[X^{S,B}_1 \indic_{A_n}] \leq \expec[Z^B_1 \,|\, Z\in [a_n, a_{n+1})] + \expec[Z^S_1]<\infty$. Similar argument also implies $\expec[X^S_1 \indic_{A_n}]<\infty$. Finally, since $N<\infty$, $p$ is bounded, Definition \ref{def: insider ad} iv) is verified using $\expec[X^B_1 \indic_{A_n}], \expec[X^S_1 \indic_{A_n}]<\infty$ for each $n\in \{1, \ldots, N\}$.
\end{proof}

\section{Convergence}\label{sec: convergence}

Collecting results from previous sections, we will prove Theorems \ref{thm: main} and \ref{thm: str conv} in this section. Let us first construct a sequence of random variables $(\tv^\delta)_{\delta>0}$, each of which will be the fundamental value in the Glosten Milgrom model with order size $\delta$.

Adding to the sequence of canonical spaces $(\Omega^\delta, \F^{Z, \delta}, (\F^{Z, \delta}_t)_{t\in [0,1]}, \prob^\delta)$, defined at the beginning of Section \ref{subsec: nonexistence}, we introduce $(\Omega^0, \F^0, (\F^0_t)_{t\in [0,1]}, \prob^0)$, where $\Omega^0= \mathbb{D}([0,1], \Real)$ is the space of $\Real$-valued \cadlag\, functions on $[0,1]$ with coordinate process $Z^0$, and $\prob^0$ is the Wiener measure. Denote by $\prob^{0,y}$ the Wiener measure under which $Z^0_0=y$ a.s.. Let us now define a $\Real\cup \{-\infty, \infty\}$-valued sequence $(a^0_n)_{n=1, \cdots N+1}$ via
\[
 a^0_1 = -\infty, \quad a^0_n = \Phi^{-1}\pare{p_1 + \cdots + p_{n-1}}, \quad n=2, \cdots, N+1, \quad \text{where } \Phi(\cdot) = \int_{-\infty}^\cdot \frac{1}{\sqrt{2\pi}} e^{-x^2/2} \,dx.
\]
Using this sequence, one can define a pricing rule following the same recipe in \eqref{eq: p func}:
\begin{align}
 &p^0(y,t) := \sum_{n=1}^N v_n h^0_n(y,t), \quad y\in \Real, t\in[0,1], n\in\{1, \cdots, N\}, \label{eq: p0}\\
 &\text{where } h^0_n(y,t):= \prob^{0,y}\pare{Z^0_{1-t}\in [a^0_n, a^0_{n+1})}= \Phi(a^0_{n+1}-y) - \Phi(a_n^0-y). \nonumber
\end{align}
As we will see later, this is exactly the pricing rule in the Kyle-Back equilibrium. Moreover, the sequence $(a^\delta_n)_{n=1, \cdots, N+1}$, associated to $(\tv^\delta)_{\delta>0}$ constructed below, converges to $(a^0_n)_{n=1, \cdots, N+1}$ as $\delta \downarrow 0$, helping to verify Definition \ref{def: Asy GM eq} i).

\begin{lem}\label{lem: v delta}
 For any $\tv$ with distribution \eqref{eq: v dist} where $N$ may not be finite, there exists a sequence of random variables $(\tv^\delta)_{\delta>0}$, each of which takes value in $\{v_1, \cdots, v_N\}$, such that
 \begin{enumerate}
  \item[i)] Assumption \ref{ass: middle level} is satisfied when $\tv$ therein is replaced by each $\tv^\delta$\,\footnote{When the order size is $\delta$, Assumption \ref{ass: middle level} iii) reads $(a_n^\delta + a_{n+1}^\delta -\delta)/2\notin \delta \Integer$.};
  \item[ii)] $Law(\tv^\delta) \Longrightarrow Law(\tv)$, as $\delta \downarrow 0$. Here $\Longrightarrow$ represents the weak convergence of probability measures.
 \end{enumerate}
\end{lem}

\begin{proof}
 For each $\delta>0$, $\tv^\delta$ will be constructed by adjusting $p_n$ in \eqref{eq: v dist} to some $p_n^\delta$, $n=1, \cdots, N$. Starting from $[\tv=v_1]$, choose $a^\delta_1=-\infty$, $a^\delta_2= \inf\{y\in \delta \Integer\,:\, \prob^\delta(Z^\delta_1 \leq y) \geq p_1\}$, and set $\prob^\delta(\tv^\delta = v_1) = \prob^\delta(Z^\delta_1 \in [a^\delta_1, a^\delta_2))$. Moving on to $[\tv^\delta = v_2]$, choose $a^\delta_3 = \inf\{y\in \delta \Integer\,:\, \prob^\delta(Z^\delta_1 \leq y) \geq p_1+p_2 \text{ and } (a^\delta_2 + y-\delta)/2 \notin \delta \Integer\}$ and  set $\prob^\delta (\tv^\delta = v_2) = \prob^\delta (Z^\delta_1 \in [a^\delta_2, a^\delta_3))$. Following this step, we can define $a^\delta_n$ inductively. When $N<\infty$, we set $a^\delta_{N+1}= \infty$. This construction gives a sequence of random variables $(\tv^\delta)_{\delta>0}$ taking values in $\{v_1, \cdots, v_N\}$ such that
 $\prob^\delta(\tv^\delta=v_n) = p^\delta_n := \prob^\delta(Z^\delta_1 \in [a^\delta_n, a^\delta_{n+1}))$ with $\sum_{n=1}^N p^\delta_n=1$, moreover each sequence $(a^\delta_n)_{n=1, \cdots, N+1}$ satisfies Assumption \ref{ass: middle level}.

 It remains to show $Law(\tv^\delta) \Longrightarrow Law(\tv)$ as $\delta\downarrow 0$. To this end, note that $a^\delta_{n}$ is either the $(\sum_{i=1}^{n-1} p_i)^{th}$ quantile of the distribution of $Z^\delta_1$ or $\delta$ above this quantile. When $\beta^\delta$ is chosen as $1/(2\delta^2)$, it follows from \cite[Chapter 6, Theorem 5.4]{Ethier-Kurtz} that $\prob^\delta \Longrightarrow \prob^0$, in particular, $Law(Z^\delta_1) \Longrightarrow Law(Z^0_1)$. Therefore,
 \begin{equation}\label{eq: a conv}
  \lim_{\delta \downarrow 0} a^\delta_n = a^0_n, \quad n=1, \cdots, N+1.
 \end{equation}
 For any $\epsilon>0$ and $n\in\{1, \cdots, N\}$, the previous convergence yields the existence of a sufficiently small $\delta_{\epsilon, n}$, such that $[a^0_n + \epsilon,a^0_{n+1}-\epsilon) \subseteq [a^\delta_n ,a_{n+1}^\delta) \subseteq [a^0_n-\epsilon,a^0_{n+1}+\epsilon)$ for any $\delta \leq \delta_{\epsilon, n}$. Hence
 \begin{align*}
	\prob^\delta\pare{Z_1^\delta \in [a^\delta_n,a^\delta_{n+1})} &\leq 	 \prob^\delta\pare{Z_1^\delta \in [a^0_n-\epsilon,a^0_{n+1}+\epsilon)} \to \prob^0\pare{Z_1^0 \in [a^0_n-\epsilon, a^0_{n+1}-\epsilon)},\\
	\prob^\delta\pare{Z_1^\delta \in [a^\delta_n,a^\delta_{n+1})} &\geq 	 \prob^\delta\pare{Z_1^\delta \in [a^0_n+\epsilon,a^0_{n+1}-\epsilon)} \to \prob^0\pare{Z_1^0 \in [a^0_n+\epsilon, a^0_{n+1}-\epsilon)}, \quad \text{as } \delta \downarrow 0,
\end{align*}
where both convergence follow from $Law(Z^\delta_1) \Longrightarrow Law(Z^0_1)$ and the fact that the distribution of $Z^0_1$ is continuous. Since $\epsilon$ is arbitrarily chosen, utilizing the continuity of the distribution for $Z^0_1$ again, we obtain from the previous two inequalities
\begin{equation*}
\lim_{\delta \downarrow 0}\prob^\delta\pare{Z_1^\delta \in [a^\delta_n,a_{n+1}^\delta)} = \prob^0\pare{Z_1^0 \in [a^0_n,a^0_{n+1})}.
\end{equation*}
Hence $\lim_{\delta \downarrow 0} p^\delta_n = p^0_n$ for each $n\in \{1, \cdots N\}$ and $Law(\tilde{v}^\delta) \Rightarrow Law(\tilde{v})$.
\end{proof}

After $(\tv^\delta)_{\delta >0}$ is constructed, it follows from Sections \ref{sec: suboptimal} and \ref{sec: point process bridge} that a sequence of strategies $(X^{B,\delta}, X^{S, \delta}; \F^{I,\delta})_{\delta>0}$ exists, each of which satisfies conditions in Proposition \ref{prop: rational}. Hence $p^\delta$ in \eqref{eq: p func} is rational for each $\delta>0$. It then remain to verify Definition \ref{def: Asy GM eq} iii) to establish an asymptotic Glosten Milgrom equilibrium.

Before doing this, we prove Theorem \ref{thm: str conv} first.
Let us recall the Kyle-Back equilibrium. Following arguments in \cite{Kyle} and \cite{Back}, the equilibrium pricing rule is given by \eqref{eq: p0} and the equilibrium demand satisfies the SDE
\[
 Y^0= Z^0 + \sum_{n=1}^N \indic_{\{\tv= v_n\}} \int_0^\cdot \frac{\partial_y h^0_n(Y^0_r, r)}{h^0_n(Y^0_r, r)}\, dr,
\]
where $Z^0$ is a $\prob^0$-Brownian motion modeling the demand from noise traders. Hence the insider's strategy in the Kyle-Back equilibrium is given by
\[
 X^0 = \sum_{n=1}^N \indic_{\{\tv= v_n\}} \int_0^\cdot \frac{\partial_y h^0_n(Y^0_r, r)}{h^0_n(Y^0_r, r)}\, dr.
\]

\begin{proof}[Proof of Theorem \ref{thm: str conv}]
 As we have seen in Lemma \ref{lem: v delta}, Assumption \ref{ass: middle level} is satisfied by each $\tv^\delta$. It then follows from Proposition \ref{prop: Y property} i) and ii) that the distribution of $Y^\delta$ on $[\tv^\delta = v_n]$ is the same as the distribution of $Z^\delta$ conditioned on $Z^\delta_1 \in [a^\delta_n, a^\delta_{n+1})$. Denote $Y^{0,n}=Y^0 \indic_{\{\tv=v_n\}}$ as the cumulative demand in Kyle Back equilibrium when the fundamental value is $v_n$. The same argument as in \cite[Lemma 5.4]{Cetin-Xing} yields
 \[
  Law(Z^\delta \,|\, Z^\delta_1 \in [a^\delta_n, a^\delta_{n+1})) \Longrightarrow Law(Y^{0,n}), \quad \text{ as } \delta \downarrow 0,
 \]
 for each $n\in \{1, \cdots, N\}$. It then follows
 \begin{equation}\label{eq: Y conv}
  Law(Y^\delta; \F^{I, \delta}) \Longrightarrow Law(Y^0; \F^{I,0}), \quad \text{ as } \delta \downarrow 0,
 \end{equation}
 where the filtration $\F^{I,0}$ is $\F^{0}$ initially enlarged by $\tv$.
 Recall from \eqref{eq: Y decomp} that $Y^\delta = Z^\delta + X^{B, \delta} - X^{S, \delta}$, moreover $Y^{0}= Z^0 + X^0$. Combining \eqref{eq: Y conv} with $Law(Z^\delta) \Longrightarrow Law(Z^0)$, we conclude from \cite[Proposition VI.1.23]{Jacod-Shiryaev} that $Law(X^{B,\delta}- X^{S,\delta})\Longrightarrow Law(X^0)$ as $\delta \downarrow 0$.
\end{proof}

In the rest of the section, Definition \ref{def: Asy GM eq} iii) is verified for strategies $(X^{B,\delta}, X^{S,\delta}; \F^{I,\delta})_{\delta>0}$, which concludes the proof of Theorem \ref{thm: main}. We have seen in Proposition \ref{prop: ep<U-L} that the expected profit of the strategy $(X^{B, \delta}, X^{S, \delta}; \F^{I,\delta})$, constructed in Section \ref{sec: point process bridge}, satisfies
\[
 \mathcal{J}^\delta(v_n, 0,0; X^{B, \delta}, X^{S, \delta}) = U^\delta(v_n, 0, 0) - L^\delta (v_n, 0, 0), \quad n\in \{1, \cdots, N\},
\]
where
\begin{equation}\label{eq: L delta}
\begin{split}
 L^\delta(v_n, 0,0) =& \quad \delta\beta^\delta\,\expec^{\delta, 0}\bra{\left.  \int_0^1 (v_n - p^\delta (\dmd, r)) \,\indic_{\{Y^\delta_{r-} = \umd\}} dr\right| \tv^\delta = v_n}\\
 &-\delta \beta^\delta\,\expec^{\delta, 0}\bra{\left. \int_0^1 (v_n - p^\delta (\umd, r)) \,\indic_{\{Y^\delta_{r-} = \dmd\}} dr\right| \tv^\delta = v_n}.
\end{split}
\end{equation}
This expression for $L^\delta$ follows from changing the order size in \eqref{eq: L} from $1$ to $\delta$ and utilizing $\theta^{B, S, \delta}(\umd, \cdot) = \theta^{S,S,\delta}(\umd, \cdot) = \theta^{S, B,\delta}(\dmd, \cdot) = \theta^{B,B,\delta}(\dmd, \cdot) =0$ from Corollary \ref{Cor: X intensity} i) and iii), $\theta^{B,T,\delta}=\theta^{S,T,\delta}\equiv 0$ from Remark \ref{rem: no top up}, and the expectations are taken under $\prob^{\delta, 0}$. Here $\dmd:= \delta\lfloor (a_n+ a_{n+1}-\delta)/2\delta\rfloor$ the largest integer multiple of $\delta$ smaller than $m^\delta_n$ and by $\umd:= \delta\lceil (a_n+ a_{n+1}-\delta)/2\delta \rceil$ the smallest integer multiple of $\delta$ larger than $m^\delta_n$.
To prove Theorem \ref{thm: str conv}, let us first show
\begin{equation}\label{eq: L conv}
 \lim_{\delta\downarrow 0} L^\delta(v_n, 0, 0) = 0, \quad n\in\{1, \cdots, N\}.
\end{equation}
In the rest development, we fix $v_n$ and denote $L^\delta = L^\delta (v_n, 0,0)$.

Before presenting technical proofs for \eqref{eq: L conv}, let us first introduce a heuristic argument. First, since $\beta^\delta= 1/(2\delta^2)$, \eqref{eq: L delta} can be rewritten as \begin{equation}\label{eq: L delta reform}
 L^\delta = \expec^{\delta,0}\bra{\left. \uId_1 \right| \tv^\delta = v_n} - \expec^{\delta,0}\bra{\left. \dId_1 \right| \tv^\delta = v_n},
\end{equation}
where
\[
 \uId_\cdot = \int_0^\cdot (v_n - p^\delta (Y^\delta_{r-}-\delta, r)) \, d \uLd_r, \quad \dId_\cdot = \int_0^\cdot (v_n - p^\delta (Y^\delta_{r-}+\delta, r))\, d \dLd_r,
\]
and $\cL^{\delta, y}_\cdot = \frac{1}{2\delta} \int_0^\cdot \indic_{\{Y^\delta_{r-} = y\}} dr$ is the \emph{scaled occupation time} of $Y^\delta$ at level $y$. Here $Y^\delta$ is, in its natural filtration, the difference of two independent Poisson $Y^{B, \delta}$ and $Y^{S, \delta}$ with jump size $\delta$ and intensity $\beta^\delta$, cf. Proposition \ref{prop: Y property} ii). For the integrands in $\uId$ and $\dId$, we expect that $v_n -p^\delta (Y^\delta_\cdot \pm \delta, \cdot) \wc v_n - p^0(Y^0_\cdot, \cdot)$, where $Y^0$ is a $\prob^0$-Brownian motion. As for the integrators, we will show both $\uLd_\cdot$ and $\dLd_\cdot$ converge weakly to $\cL^{m_n}_\cdot$, which is the Brownian local time at level $m_n := (a^0_n + a^0_{n+1})/2$. Then the weak convergence of both integrands and integrators yield
\[
 \uId_\cdot \,\text{ and } \, \dId_\cdot \wc I^{0,n}_\cdot:= \int_0^\cdot (v_n - p^0(Y^0_r, r)) \, d\cL^{m_n}_r, \quad \text{ as } \delta \downarrow 0.
\]
Finally passing the previous convergence to conditional expectation, the two terms on the right hand side of \eqref{eq: L delta reform} cancel each other in the limit.

To make this heuristic argument rigorous, let us first prepare several results.

\begin{prop}\label{prop: integrand conv}
 On the family of filtration $(\F^{Y, \delta}_t)_{t\in [0,1], \delta \geq 0}$, generated by $(Y^\delta)_{\delta \geq 0}$,
 \[
  p^\delta (Y^\delta_\cdot \pm \delta, \cdot) \wc p^0(Y^0_\cdot, \cdot) \quad \text{ on } \mathbb{D}[0,1) \text{ as  } \delta \downarrow 0.
 \]
\end{prop}
\begin{proof}
 To simplify presentation, we will prove
 \begin{equation}\label{eq: pdY conv}
  p^\delta (Y^\delta_\cdot, \cdot) \wc p^0(Y^0_\cdot, \cdot) \quad \text{ as } \delta \downarrow 0.
 \end{equation}
 The assertions with $\pm \delta$ can be proved by replacing $Y^\delta$ by $Y^\delta \pm \delta$. First, applying It\^{o}'s formula and utilizing \eqref{eq: p pde} yield
 \begin{equation}\label{eq: pdY ito}
 \begin{split}
  p^\delta(Y^\delta_\cdot, \cdot) = p^\delta(0,0) &+ \int_0^\cdot \frac{1}{\delta} \pare{p^\delta(Y^\delta_{r-} + \delta, r) -p^\delta(Y^\delta_{r-}, r)} \, d\overline{Y}^{B, \delta}_r\\
  &+ \int_0^\cdot \frac{1}{\delta} \pare{p^\delta(Y^\delta_{r-} - \delta, r) -p^\delta(Y^\delta_{r-}, r)} \, d\overline{Y}^{S, \delta}_r,
 \end{split}
 \end{equation}
 where $\overline{Y}^{B,\delta}_\cdot = Y^{B,\delta}_\cdot - \delta \beta^\delta \cdot$ and $\overline{Y}^{S,\delta}_\cdot = Y^{S,\delta}_\cdot - \delta \beta^\delta \cdot$ are compensated jump processes.
 For $p^\delta(0,0)$ on the right hand side, the same argument in Lemma \ref{lem: v delta} yields $\lim_{\delta\downarrow 0}p^\delta(0,0) =  p^0(0,0)$. As for the other two stochastic integrals, we will show that they converge weakly to
 \[
  \frac{1}{\sqrt{2}} \int_0^\cdot \partial_y p^0(Y^0_r,r) dW^B_r \quad \text{ and } \quad -\frac{1}{\sqrt{2}} \int_0^\cdot \partial_y p^0(Y^0_r, r) dW^S_r, \quad \text{respectively},
 \]
 where $W^B$ and $W^S$ are two independent Brownian motion. These estimates then imply the right hand side of \eqref{eq: pdY ito} converges weakly to
 \[
  p^0(0,0) + \int_0^\cdot \partial_y p^0(Y^0_r, r) \, dW_r,
 \]
 where $W=W^B/\sqrt{2} - W^S/\sqrt{2}$ is another Brownian motion. Since $p^0$ satisfies $\partial_t p^0 + \frac12 \partial^2_{yy} p^0 =0$, the previous process has the same law as $p^0(Y^0_\cdot, \cdot)$. Therefore \eqref{eq: pdY conv} is confirmed.

 To prove the aforementioned convergence of stochastic integrals, let us first derive the convergence of $(p^\delta(\cdot+\delta, \cdot) - p^\delta(\cdot, \cdot))/\delta$ on $\Real \times [0,1)$. To this end, it follows from \eqref{eq: p func} that
 \begin{align*}
  &\frac{1}{\delta}  (p^\delta(y+\delta, t) - p^\delta(y,t))\\
  &= \frac{1}{\delta} \sum_{n=1}^N v_n \bra{\prob^{\delta, y+\delta}(Z_{1-t}^\delta \in [a^\delta_n, a^\delta_{n+1})) - \prob^{\delta, y}(Z^\delta_{1-t} \in [a^\delta_n, a^\delta_{n+1}))} \\
  &= \frac{1}{\delta} \sum_{n=1}^N v_n \bra{\prob^{\delta, y}(Z^{\delta}_{1-t}=a_n^\delta-\delta) - \prob^{\delta, y}(Z^{\delta}_{1-t} = a^\delta_{n+1} -\delta)}\\
  &= \frac{1}{\delta} \sum_{n=1}^N v_n \bra{\prob^{1, 0}\pare{Z^{1}_{1-t} =\frac{a^\delta_n-\delta-y}{\delta}} - \prob^{1,0}\pare{Z^1_{1-t}=\frac{a^\delta_{n+1}-\delta-y}{\delta}}}\\
  &= \sum_{n=1}^N v_n \bra{\frac1\delta e^{-\frac{1-t}{\delta^2}} I_{\left|\frac{a^\delta_n-\delta-y}{\delta}\right|}\pare{\frac{1-t}{\delta^2}} - \frac1\delta e^{-\frac{1-t}{\delta^2}} I_{\left|\frac{a^\delta_{n+1}-\delta-y}{\delta}\right|}\pare{\frac{1-t}{\delta^2}} }\\
  & \rightarrow \sum_{n=1}^N v_n \bra{\frac{1}{\sqrt{2\pi (1-t)}} \exp\pare{-\frac{(a_n^0-y)^2}{2(1-t)}} - \frac{1}{\sqrt{2\pi (1-t)}} \exp\pare{-\frac{(a_{n+1}^0-y)^2}{2(1-t)}}}\\
  &= \partial_y p^0(y,t), \quad \text{ as } \delta \downarrow 0.
 \end{align*}
Here $Z^1_{1-t}$ is the difference of two independent Poisson random variables with common parameter $(1-t)\beta^\delta= (1-t)(2\delta^2)^{-1}$ under $\prob^{1,0}$. Hence the fourth identity above follows from the probability distribution function of the \emph{Skellam distribution}: $\prob^{1,0}(Z_{1-t}^1=k)= e^{-2\mu} I_{|k|}(2\mu)$, where $I_{|k|}(\cdot)$ is the \emph{modified Bessel function of the second kind} and $\mu=(1-t)(2\delta^2)^{-1}$, cf. \cite{Skellam}. The convergence above is locally uniformly in $\Real\times[0,1)$ according to \cite[Theorem 2]{Athreya}. The last identity above follows from taking $y$ derivative to $p^0(y,t) = \sum_{n=1}^N \pare{\Phi\pare{\frac{a^0_{n+1}-y}{\sqrt{1-t}}} - \Phi\pare{\frac{a^0_n -y}{\sqrt{1-t}}}}$, cf. \eqref{eq: p0}. Combining the previous locally uniform convergence of $(p^\delta(\cdot+\delta, \cdot) - p^\delta(\cdot, \cdot))/\delta$ with the weak convergence $Y^\delta \wc Y^0$ in their natural filtration, we have from \cite[Chapter 1, Theorem 5.5]{Billingsley-conv}:
\[
 \frac{1}{\delta} \pare{p^\delta(Y^\delta_\cdot + \delta, \cdot) - p^\delta(Y^\delta_\cdot, \cdot)} \wc \partial_y p^0(Y^0_\cdot, \cdot) \quad \text{ on } \mathbb{D}[0,1)  \text{ as } \delta \downarrow 0.
\]

As for the integrators in \eqref{eq: pdY ito},  $\overline{Y}^{B, \delta} \wc W^B/\sqrt{2}$ and $\overline{Y}^{S, \delta} \wc W^S/\sqrt{2}$. Moreover, both $(\overline{Y}^{B, \delta})_{\delta>0}$ and $(\overline{Y}^{S, \delta})_{\delta>0}$ are \emph{predictable uniform tight} (P-UT), since $\langle \overline{Y}^{B, \delta}\rangle_t = \langle \overline{Y}^{S, \delta}\rangle_t =t/2$, for any $\delta>0$, cf. \cite[Chapter VI, Theorem 6.13 (iii)]{Jacod-Shiryaev}. Then
combining weak convergence of both integrands and integrators, we obtain from \cite[Chapter VI, Theorem 6.22]{Jacod-Shiryaev} that
\[
 \int_0^\cdot \frac{1}{\delta} (p^\delta(Y^\delta_{r-} +\delta, r) - p^\delta(Y^\delta_{r-}, r)) \,d \overline{Y}^{B,\delta}_r \wc \frac{1}{\sqrt{2}}\int_0^\cdot \partial_y p^0(Y^0_r, r) \, dW^B_r \quad \text{ on } \mathbb{D}[0,1) \text{ as } \delta \downarrow 0.
\]
A similar weak convergence holds for the other stochastic integral in \eqref{eq: pdY ito} as well. Therefore the claimed weak convergence of stochastic integrals on the right hand side of \eqref{eq: pdY ito} is confirmed.
\end{proof}

Having studied the weak convergence of integrands in $\uId$ and $\dId$, let us switch our attention to the integrators $\uLd$ and $\dLd$.

\begin{prop}\label{prop: local time}
 On the family of filtration $(\F^{Y, \delta}_t)_{t\in[0,1], \delta\geq 0}$, for any $n\in \{1, \cdots, N\}$,
 \[
  \uLd \wc \cL^{m_n} \quad \text{and} \quad \dLd \wc \cL^{m_n} \quad \text{ on } \mathbb{D}[0,1] \text{ as } \delta\downarrow 0.
 \]
\end{prop}

\begin{proof}
 For simplicity of presentation, we will prove
 \begin{equation}\label{eq: lt conv}
  \cL^{\delta, 0} \wc \cL^0 \quad \text{ as } \delta \downarrow 0.
 \end{equation}
 Since $\lim_{\delta\downarrow 0} \umd=\lim_{\delta\downarrow 0} \dmd=m_n$ follows from \eqref{eq: a conv},  the statement of the proposition follows from replacing $Y^\delta$ by $Y^\delta -\umd$ (or by $Y^\delta - \dmd$) and $Y^0$ by $Y^0-m_n$ in the rest of the proof. To prove \eqref{eq: lt conv}, applying It\^{o}'s formula to $|Y^\delta_\cdot|$ yields
 \begin{equation}\label{eq: Y abs}
 \begin{split}
  |Y^\delta_\cdot| =& \sum_{r\leq \cdot} \pare{|Y^\delta_r| - |Y^\delta_{r-}|} \\
  =& \int_0^\cdot \pare{|Y^\delta_{r-} + \delta| - |Y^\delta_{r-}|} \,d(Y^{B, \delta}_r/\delta - \beta^\delta r)  + \int_0^\cdot \pare{|Y^\delta_{r-} - \delta| - |Y^\delta_{r-}|} \,d(Y^{S, \delta}_r/\delta - \beta^\delta r)\\
  &+ \int_0^\cdot \pare{|Y^\delta_{r-}+\delta| + |Y^\delta_{r-}-\delta| - 2|Y^\delta_{r-}|}\beta^\delta dr\\
  =& \int_0^\cdot \pare{|Y^\delta_{r-} + \delta| - |Y^\delta_{r-}|} \,d\overline{Y}^{B, \delta}_r/\delta  + \int_0^\cdot \pare{|Y^\delta_{r-} - \delta| - |Y^\delta_{r-}|} \,d\overline{Y}^{S, \delta}_r/\delta + \int_0^\cdot\frac{1}{\delta} \indic_{\{Y^\delta_{r-} =0\}} dr,
 \end{split}
 \end{equation}
 where the third identity follows from $|y+\delta|+|y-\delta| -2|y|=2\delta\, \indic_{\{y=0\}}$ for any $y\in \Real$.
 On the other hand, Tanaka formula for Brownian motion is
 \begin{equation}\label{eq: tanaka}
  |Y^0_\cdot| = \int_0^\cdot sgn(Y^0_r) \,dY^0_r + 2 \cL^0_\cdot,
 \end{equation}
 where $sgn(x) = 1$ when $x>0$ or $-1$ when $x\leq 0$.

 The convergence \eqref{eq: lt conv} is then confirmed by comparing both sides of \eqref{eq: Y abs} and \eqref{eq: tanaka}. To this end, since $Y^\delta \wc Y^0$ and the absolute value is a continuous function, then $|Y^\delta|\wc |Y^0|$ follows from \cite[Chapter 1, Theorem 5.1]{Billingsley-conv}. Then \eqref{eq: lt conv} is confirmed as soon as we prove the martingale term on the right hand side of \eqref{eq: Y abs} converges weakly to the martingale in \eqref{eq: tanaka}, which we prove in the next result.
\end{proof}

\begin{lem}\label{lem: mart conv}
 Let $M^\delta:= \int_0^\cdot \pare{|Y^\delta_{r-} + \delta| - |Y^\delta_{r-}|} \,d\overline{Y}^{B, \delta}_r/\delta  + \int_0^\cdot \pare{|Y^\delta_{r-} - \delta| - |Y^\delta_{r-}|} \,d\overline{Y}^{S, \delta}_r/\delta$ and $M^0:= \int_0^\cdot sgn(Y^0_r) \,dY^0_r$. Then $M^\delta \wc M^0$ on $\mathbb{D}[0,1]$ as $\delta \downarrow 0$.
\end{lem}

\begin{proof}
 Define $f^\delta(y):= \frac{1}{\delta}(|y+\delta|- |y|)$ for $y\in \Real$ and observe
 \[
   f^\delta(y) = \left\{\begin{array}{ll}1 & y\geq 0\\ 2y/\delta+1 & -\delta<y<0 \\ -1 & y\leq -\delta \end{array}\right..
 \]
 It is clear that $f^\delta$ converges to $sgn(\cdot)$ locally uniformly on $\Real\setminus\{0\}$. On the other hand, $Y^\delta \wc Y^0$ and the law of $Y^0$ is continuous. It then follows from \cite[Chapter 1, Theorem 5.5]{Billingsley-conv} that $f^\delta({Y^\delta}) \wc sgn(Y^0)$. As for the integrators $(\overline{Y}^{B,\delta})_{\delta > 0}$, as we have seen in the proof of Proposition \ref{prop: integrand conv}, they converge weakly to $W^B/\sqrt{2}$ and are P-UT. Then \cite[Chapter VI, Theorem 6.22]{Jacod-Shiryaev} implies
 \[
 \int_0^\cdot \pare{|Y^\delta_{r-} + \delta| - |Y^\delta_{r-}|} \,d\overline{Y}^{B, \delta}_r/\delta \wc \frac{1}{\sqrt{2}} \int_0^\cdot sgn(Y^0_r) \,dW^B_r.
 \]
 Similar argument yields
 \[
  \int_0^\cdot \pare{|Y^\delta_{r-} - \delta| - |Y^\delta_{r-}|} \,d\overline{Y}^{S, \delta}_r/\delta \wc -\frac{1}{\sqrt{2}} \int_0^\cdot sgn(Y^0_r) \,dW^S_r.
 \]
 Here $W^B$ and $W^S$ are independent Brownian motion. Defining $W=W^B/\sqrt{2} - W^S/\sqrt{2}$, we obtain from the previous two convergence that
 \[
  M^\delta \wc \int_0^\cdot sgn(Y^0_r) \,dW_r \quad \text{which has the same law as} \quad M^0.
 \]
\end{proof}

Propositions \ref{prop: integrand conv} and \ref{prop: local time} combined yields the weak convergence of $(\uId)_{\delta>0}$ and $(\dId)_{\delta>0}$. Moreover the sequence of local time in Proposition \ref{prop: local time} also converge in expectation.

\begin{cor}\label{cor: I conv}
 On the family of filtration $(\F^{Y,\delta}_t)_{t\in[0,1], \delta \geq 0}$, for any $n\in\{1, \cdots, N\}$,
 \[
  \uId \text{ and } \dId \wc I^{0,n} \quad \text{ on } \mathbb{D}[0,1) \text{ as } \delta \downarrow 0.
 \]
\end{cor}

\begin{proof}
 The statement follows from combining Propositions \ref{prop: integrand conv} and \ref{prop: local time}, and appealing to \cite[Chapter VI, Theorem 6.22]{Jacod-Shiryaev}. In order to apply the previous result, we need to show that both $(\uLd)_{\delta >0}$ and $(\dLd)_{\delta >0}$ are P-UT. This property will be verified for $(\uLd)_{\delta >0}$. The same argument works for $(\dLd)_{\delta >0}$ as well. To this end, since $\uLd$ is a nondecreasing process, $(\uLd)_{\delta >0}$ is P-UT as soon as $(Var(\uLd)_1)_{\delta >0}$ is tight, where $Var(X)$ is the variation of the process $X$, cf. \cite[Chapter VI, 6.6]{Jacod-Shiryaev}. Note $Var(\uLd)_1=\uLd_1$, since $\uLd$ is nondecreasing. Then the tightness of $(Var(\uLd)_1)_{\delta >0}$ is implied by Proposition \ref{prop: local time}.
\end{proof}

\begin{cor}\label{cor: local time exp}
For any  $n\in\{1, \cdots, N\}$ and $t\in[0,1]$,
 \[
  \lim_{\delta \downarrow 0} \expec^{\delta, 0} \bra{\uLd_t} = \lim_{\delta \downarrow 0} \expec^{\delta, 0} \bra{\dLd_t} = \expec^{0,0}\bra{\cL^{m_n}_t}.
 \]
\end{cor}

\begin{proof}
 For simplicity of presentation, we will prove $\lim_{\delta \downarrow 0} \expec^{\delta, 0}[\cL^{\delta, 0}_t]=\expec^{0,0}[\cL^0_t]$. Then the statement of the corollary follows from replacing $Y^\delta_t$ by $Y^{\delta}_t - \umd$ or $Y^\delta_t -\dmd$ in the rest of the proof. Since the stochastic integrals in \eqref{eq: Y abs} are $\prob^{\delta, 0}$-martingales,
 \[
  2\expec^{\delta, 0}[\cL^{\delta, 0}_t] = \expec^{\delta, 0}[|Y^\delta_t|].
 \]
 Since $\expec[(Y^\delta_t)^2] = t$ for any $\delta>0$, $(|Y^\delta_t|; \prob^{\delta, 0})_{\delta >0}$ is uniformly integrable. It then follows from \cite[Appendix, Proposition 2.3]{Ethier-Kurtz} and $Law(|Y^\delta_t|) \Longrightarrow Law(|Y^0_t|)$ that $\lim_{\delta\downarrow 0} \expec^{\delta, 0}[|Y^\delta_t|] = \expec^{0,0}[|Y^0_t|]$. Therefore the claim follows since $\expec^{0,0}[|Y^0_t|]=2\expec^{0,0}[\cL^0_t]$ cf. \eqref{eq: tanaka}.
\end{proof}

Collecting previous results, the following result confirms \eqref{eq: L conv}.

\begin{prop}\label{prop: L conv}
 For the strategies $(X^{B, \delta}, X^{S, \delta}; \F^{I, \delta})_{\delta>0}$ constructed in Section \ref{sec: point process bridge},
 \[
  \lim_{\delta\downarrow 0} L^\delta(v_n, 0, 0) =0, \quad n\in \{1, \cdots, N\}.
 \]
\end{prop}

\begin{proof}
 Fix any $\epsilon\in(0,1)$. Corollary \ref{cor: I conv} implies that $Law(\uId_{1-\epsilon}; \F^{Y,\delta}) \Longrightarrow Law(I^{0,n}_{1-\epsilon}; \F^0)$. Recall $Law(\tv^\delta) \Longrightarrow Law(\tv)$ from Lemma \ref{lem: v delta}. It then follows
 \begin{equation*}
  Law\pare{\uId_{1-\epsilon} \,\indic_{\{\tv^\delta = v_n\}}; \F^{Y,\delta}} \Longrightarrow Law\pare{I^{0,n}_{1-\epsilon} \,\indic_{\{\tv=v_n\}}; \F^0}.
 \end{equation*}
 On the other hand, since $N$ is finite, $p^\delta$ is bounded uniformly in $\delta$. Then there exists constant $C$ such that $|\uId_{1-\epsilon}| \,\indic_{\{\tv^\delta = v_n\}} \leq C \uLd_{1-\epsilon}$, where the expectation of the upper bound converges, cf. Corollary \ref{cor: local time exp}. Therefore appealing to \cite[Appendix Theorem 1.2]{Ethier-Kurtz} and utilizing $\lim_{\delta\downarrow 0}\prob^{\delta}(\tv^\delta=v_n) = \prob^{0}(\tv=v_n)$ from Lemma \ref{lem: v delta}, we obtain
 \begin{equation}\label{eq: I 1-eps}
  \expec^{\delta, 0}\bra{\uId_{1-\epsilon} \,|\, \tv^\delta = v_n} = \frac{\expec^{\delta, 0} \bra{\uId_{1-\epsilon} \, \indic_{\{\tv^\delta = v_n\}}}}{\prob^\delta(\tv^\delta = v_n)} \rightarrow \frac{\expec^{0, 0} \bra{I^{0,n}_{1-\epsilon} \, \indic_{\{\tv = v_n\}}}}{\prob^0(\tv = v_n)} = \expec^{0,0}\bra{I^{0,n}_{1-\epsilon} \,|\, \tv=v_n}, \quad \text{ as } \delta \downarrow 0.
 \end{equation}

 On the other hand, since $\lim_{\delta \downarrow 0} \prob^\delta(\tv^\delta) = \prob^0(\tv=v_n)>0$, there exists a constant $C$ such that
 \[
  \expec^{\delta, 0}\bra{\left.|\uId_1- \uId_{1-\epsilon}|\right|\, \tv^\delta=v_n} \leq C \,\expec^{\delta, 0} \bra{\uLd_1 - \uLd_{1-\epsilon}} \rightarrow C\, \expec^{0,0} \bra{\cL^{m_n}_1 - \cL^{m_n}_{1-\epsilon}}, \quad \text{ as } \delta \downarrow 0,
 \]
 where the convergence follows from applying Corollary \ref{cor: local time exp} twice. For the difference of Brownian local time, L\'{e}vy's result (cf. \cite[Chapter 3, Theorem 6.17]{Karatzas-Shreve-BM}) yields
 \[
  \expec^{0,0} \bra{\cL^{m_n}_1 - \cL^{m_n}_{1-\epsilon}} = \expec^{0, -m_n}\bra{\cL^0_1 - \cL^0_{1-\epsilon}} = \frac12 \expec^{0, -m_n}\bra{\sup_{r\leq 1} Y^0_r - \sup_{r\leq 1-\epsilon} Y^0_r} = \sqrt{\frac{2}{\pi}}(1-\sqrt{1-\epsilon}),
 \]
 where $Y^0$ is a $\prob^0$-Brownian motion and $\expec^{0,y}[\sup_{r\leq t} Y^0_r] = \sqrt{2t/\pi} +y$ is utilized to obtain the third identity. Now the previous two estimates combined yield
 \begin{equation}\label{eq: I 1-eps 1}
  \limsup_{\delta \downarrow 0}\expec^{\delta, 0}\bra{\left.|\uId_1- \uId_{1-\epsilon}|\right|\, \tv^\delta=v_n} \leq C(1-\sqrt{1-\epsilon}), \quad \text{ for another constant } C.
 \end{equation}

 Estimates in \eqref{eq: I 1-eps} and \eqref{eq: I 1-eps 1} also hold when $\uId$ is replaced by $\dId$. These estimates then yield
 \[
 \begin{split}
  &\expec^{\delta, 0}\bra{\uId_1 -\dId_1 \,|\, \tv^\delta =v_n} \\
  &\leq \expec^{\delta, 0} \bra{\uId_{1-\epsilon} - \dId_{1-\epsilon} \,|\, \tv^\delta = v_n} + \expec^{\delta, 0} \bra{\left.|\uId_1 - \uId_{1-\epsilon}|\,\right|\, \tv^\delta=v_n} + \expec^{\delta, 0} \bra{\left.|\dId_1 - \dId_{1-\epsilon}|\,\right|\, \tv^\delta=v_n}.
 \end{split}
 \]
 Sending $\delta \downarrow 0$ in the previous inequality, the first term on the right side vanishes in the limit, because both conditional expectations converge to the same limit, the limit superior of both second and third terms are less than $C(1-\sqrt{1-\epsilon})$. Now since $\epsilon$ is arbitrarily choose, sending $\epsilon \rightarrow 1$ yields $\limsup_{\delta\downarrow 0} \expec^{\delta, 0}\bra{\uId_1 -\dId_1 \,|\, \tv^\delta =v_n} \leq 0$. Similar argument leads to $\liminf_{\delta\downarrow 0} \expec^{\delta, 0}\bra{\uId_1 -\dId_1 \,|\, \tv^\delta =v_n} \geq 0$, which concludes the proof.
\end{proof}

Finally the proof of Theorem \ref{thm: main} is concluded.

\begin{proof}[Proof of Theorem \ref{thm: main}]
 It remains to verify Definition \ref{def: Asy GM eq} iii).
 Fix $v_n$ and $(y,t)=(0,0)$ throughout the proof. We have seen from Proposition \ref{prop: U>V} that $V^\delta\leq U^{S,\delta}$. On the other hand, Proposition \ref{prop: ep<U-L} yields $\mathcal{J}(X^{B,\delta}, X^{S, \delta}) = U^\delta - L^\delta$. Therefore
 \[
  \sup_{(X^B, X^S) \text{ admissible}} \mathcal{J}^\delta(X^B, X^S) - \mathcal{J}^\delta(X^{B, \delta}, X^{S, \delta}) \leq U^{S,\delta} - U^\delta + L^\delta.
 \]
 Since $\lim_{\delta\downarrow 0} L^\delta =0$ is proved in Proposition \ref{prop: L conv}, it suffices to show $\lim_{\delta \downarrow 0} U^{S, \delta} - U^\delta =0$.
 To this end, from the definition of $U^{S, \delta}$,
 \begin{equation}\label{eq: US-S}
  U^{S, \delta}(0,0) - U^\delta(0,0)= (U^\delta(-\delta,0) - U^\delta(0,0))\,\indic_{\{0\leq \dmd\}}=\delta(v_n - p^\delta(0,0)) \indic_{\{0\leq \dmd\}}.
 \end{equation}
 The second identity above follows from \eqref{eq: U eq <dm} which reads $U^\delta(y,t) -U^\delta(y-1, t) + \delta(v_n -p^\delta(y,t)) =0$ for $y\leq \dmd$ when the order size is $\delta$. Therefore $\lim_{\delta\downarrow 0} U^{S, \delta} - U^\delta =0$ is confirmed after sending $\delta\downarrow 0$ in \eqref{eq: US-S}.
\end{proof}

\appendix
\section{Viscosity solutions}\label{app: vis soln}

Proposition \ref{prop: vis soln} will be proved in this section. To simplify notation, $\delta=1$ and $\tilde{v} =v_n$ are fixed throughout this section. First let us recall the definition of (discontinuous) viscosity solution to \eqref{eq: HJB}. Given a locally bounded function\footnote{Since the state space $\Integer$ is discrete, $v$ is locally bounded if $v(y,\cdot)$ is bounded in any bounded neighborhood of $t$ and any fixed $y\in \Integer$.} $v: \Integer \times [0,1]\rightarrow \Real$, its \emph{upper-semicontinuous envelope} $v^*$ and \emph{lower-semicontinuous envelope} $v_*$ are defined as
\begin{equation}\label{eq: envelope}
 v^*(y,t) := \limsup_{t'\rightarrow t} v(y,t'), \quad v_*(y,t) := \liminf_{t'\rightarrow t} v(y, t'), \quad (y,t) \in \Integer \times [0,1].
\end{equation}
\begin{defn}\label{def: vis soln}
 Let $v: \Integer \times [0,1] \rightarrow \Real$ be locally bounded.
 \begin{enumerate}[i)]
  \item $v$ is a \emph{(discontinuous) viscosity subsolution} of \eqref{eq: HJB} if
      \[
       -\varphi_t (y,t) - H(y, t, v^*) \leq 0,
      \]
      for all $y\in \Integer$, $t\in [0,1)$, and any function $\varphi: \Integer \times [0,1] \rightarrow \Real$ continuously differentiable in the second variable such that $(y,t)$ is a maximum point of $v^*-\varphi$.
  \item $v$ is a \emph{(discontinuous) viscosity supersolution} of \eqref{eq: HJB} if
      \[
       -\varphi_t (y,t) - H(y, t, v_*) \geq 0,
      \]
      for all $y\in \Integer$, $t\in [0,1)$, and any function $\varphi: \Integer \times [0,1] \rightarrow \Real$ continuously differentiable in the second variable such that $(y,t)$ is a minimum point of $v_*-\varphi$.
  \item We say that $v$ is a \emph{(discontinuous) viscosity solution} of \eqref{eq: HJB} if it is both subsolution and supersolution.
 \end{enumerate}
\end{defn}

For the insider's optimization problem, let us recall the \emph{dynamic programming principle}  (cf. e.g. \cite[Remark 3.3.3]{Pham}). Given an admissible strategy $(X^B, X^S)$, any $[t,1]-$valued stopping time $\tau$, and the fundamental value $v_n$, denote the associated profit by
\begin{align*}
 \mathcal{I}^{n}_{t, \tau} := &  \int_t^\tau (v_n-p(Y_{r-}+1,r))dX_r^{B,B} + \int_t^\tau (v_n - p(Y_{r-} + 2, r)) dX^{B,T}_r +\int_t^\tau(v_n-p(Y_{r-},r))dX_r^{B,S}\\
      &-\int_t^\tau (v_n-p(Y_{r-}-1,r))dX_r^{S,S} -\int_t^\tau (v_n - p(Y_{r-}-2, r)) dX^{S,T}_r - \int_t^\tau(v_n-p(Y_{r-},r))dX_r^{S,B},
\end{align*}
where $Y= Z+X^B-X^S$.
Then the dynamic programming principle reads:

\begin{enumerate}
 \item[DPP i)]\label{DPP i} For any admissible strategy $(X^B, X^S)$ and any $[t,1]$-valued stopping time $\tau$,
     \[
     \begin{split}
      V(y,t) \geq \expec^{y,t} \Huge[V(\tau, Y_\tau)+ \mathcal{I}^n_{t, \tau}\Huge].
     \end{split}
     \]
 \item[DPP ii)]\label{DPP ii} For any $\epsilon>0$, there exists an admissible strategy $(X^B, X^S)$ such that for all $[t, 1]$-valued stopping time $\tau$,
     \[
     \begin{split}
      V(y,t) -\epsilon \leq \expec^{y,t}\Huge[V(\tau, Y_\tau) +\mathcal{I}^n_{t, \tau} \Huge].
     \end{split}
     \]
\end{enumerate}

The viscosity solution property of the value function $V$ follows from the dynamic programming principle and standard arguments in viscosity solutions, (see e.g. \cite[Propositions 4.3.1 and 4.3.2]{Pham}\footnote{Therein the stopping time $\tau_m$ can be chosen as the first jump time of $Y$ where $Y_{t_m}=y$ for a sequence $(t_m)_m \rightarrow \overline{t}$}.) Therefore Proposition \ref{prop: vis soln} i) is verified.

\begin{rem}\label{rem: weak DPP}
 The proof of DPP ii) utilizes the measurable selection theorem. To avoid this technical result, one could employ the \emph{weak} dynamic programming principle in \cite{Bouchard-Touzi}. For the insider's optimization problem, the weak dynamic programming principle reads:
 \begin{itemize}
  \item[WDPP i)] For any $[t,1]$-valued stopping time $\tau$,
  \[V(y,t) \leq \sup_{(X^B, X^S)}\expec^{y,t}\bra{V^*(\tau, Y_\tau) + \mathcal{I}^n_{t, \tau}}.\]
  \item[WDPP ii)] For any $[t,1]$-valued stopping time $\tau$ and any upper-semicontinuous function $\varphi$ on $\Integer \times [0,1]$ such that $V\geq \varphi$, then
      \[
       V(y,t) \geq \sup_{(X^B, X^S)}\expec^{y,t}\bra{\varphi(\tau, Y_\tau) + \mathcal{I}^n_{t,\tau}}.
      \]
 \end{itemize}
 Conditions A1, A2, and A3 from Assumption A in \cite{Bouchard-Touzi} are clearly satisfied in the current context. Condition A4 from Assumption A can be verified following the same argument in \cite[Proposition 5.4]{Bouchard-Touzi}. Therefore aforementioned weak dynamic programming principle holds. Hence the value function is a viscosity solution to \eqref{eq: HJB} following from arguments similar to \cite[Section 5.2]{Bouchard-Touzi}.
\end{rem}

Now the proof of Proposition \ref{prop: vis soln} ii) is presented. To prove $(v_n, y, t, V) \in \textit{dom}(H)$, observe from the viscosity supersolution property of $V$ that $H(v_n, y, t, V_*) <\infty$, hence $(v_n, y, t, V_*)\in \text{dom}(H)$. On the other hand, for any integrable intensities $\theta^{i,j}$, $i\in \{B,S\}$ and $j\in \{B,T,S\}$, due to Definition \ref{def: insider ad} iv), one can show $\expec^{y,t}[\mathcal{I}^{n}_{t,1}]$
is a continuous function in $t$. As a supremum of a family of continuous function (cf. \eqref{eq: value func}), $V$ is then lower-semicontinuous in $t$. Therefore $V_*\equiv V$, which implies $(v_n, y, t, V) \in \text{dom}(H)$ for any $v_n$, $(y,t)\in \Integer \times [0,1)$. It then follows from \eqref{eq: buy ineq} and \eqref{eq: sell ineq} that
\begin{equation}\label{eq: ul V}
 V(y-1,t) + p(y-1, t) - v_n\leq V(y,t) \leq V(y-1,t) + p(y,t) - v_n, \quad \text{ for any } (y,t)\in \Integer \times [0,1).
\end{equation}
Taking limit supremum in $t$ in the previous inequalities and utilizing the continuity of $t\mapsto p(y,t)$, it follows that the previous inequalities still hold when $V$ is replaced by $V^*$, which means $(v_n, y, t, V^*)\in \text{dom}(H)$ for any $v_n$, $(y,t)\in \Integer \times [0,1)$.  As a result, $H(v_n, y, t, V_*)$ and $H(v_n, y, t, V^*)$ have the reduced form \eqref{eq: H red} where $V$ is replaced by $V_*$ and $V^*$, respectively. Hence Definition \ref{def: vis soln} implies that $V$ is a viscosity solution of \eqref{eq: red HJB}.

To prove Proposition \ref{prop: vis soln} iii) and iv), let us first derive a comparison result for \eqref{eq: red HJB}. The function $v: \Integer \times [0,1] \rightarrow \Real$ has at most polynomial growth in its first variable if there exist $C$ and $n$ such that $|v(y,t)| \leq C(1+|y|^n)$, for any $(y,t)\in \Integer \times [0,1]$.

\begin{lem}\label{lem: vis comparison}
Assume that $u$ (resp. $v$) has at most polynomial growth and that it is  upper-semicontinuous viscosity subsolution (resp. lower-semicontinuous supersolution) to \eqref{eq: red HJB}. If $u(\cdot, 1)\leq v(\cdot, 1)$, then $u\leq v$ in $\Integer \times [0,1)$.
\end{lem}

Assume this comparison result for a moment. Inequalities \eqref{eq: ul V} and Assumption \ref{ass: p poly} combined imply that $V$ is of at most polynomial growth. Then Lemma \ref{lem: vis comparison} and \eqref{eq: envelope} combined yield $V_*\leq V^* \leq V_*$, which implies the continuity of $t\mapsto V(y,t)$, hence Proposition \ref{prop: vis soln} iii) is verified. On the other hand, one can prove $\tilde{V}(y,t) := \expec^{y,t} \bra{V(Z_1, 1)}$ is of at most polynomial growth and is another viscosity solution to \eqref{eq: red HJB}\footnote{Write $\tilde{V}(y,t)= \expec^0\bra{V(Z_{1-t}+y, 1)}$. One can utilize the Markov property of $Z$ to show that $\tilde{V}$ is continuous differentiable and $\tilde{V}$ is a classical solution to \eqref{eq: red HJB}.}.  Then Lemma \ref{lem: vis comparison} yields
\[
 V(y,t) = \tilde{V}(y,t) = \expec^{y,t}\bra{V(Z_1,1)},
\]
which confirms Proposition \ref{prop: vis soln} iv) via the Markov property of $Z$.

\begin{proof}[Proof of Lemma \ref{lem: vis comparison}]
 For $\lambda>0$, define $\tilde{u}=e^{\lambda t} u$ and $\tilde{v} = e^{\lambda t} v$. One can check $\tilde{u}$ (resp. $\tilde{v}$) is a viscosity subsolution (resp. supersolution) to
 \begin{equation}\label{eq: HJB lambda}
  -w_t + \lambda w - \pare{w(y+1, t) - 2w(y, t) + w(y-1, t)} \beta =0.
 \end{equation}
 Since the comparison result for \eqref{eq: HJB lambda} implies the comparison result for \eqref{eq: HJB}, it suffices to consider $u$ (resp. $v$) as the viscosity subsolution (resp. supersolution) of \eqref{eq: HJB lambda}.

 Let $C$ and $n$ be constants such that $|u|, |v| \leq C(1+ |y|^n)$ on $\Integer \times [0,1]$. Consider $\psi(y,t) = e^{-\alpha t}(y^{2n} + \tilde{C})$ for some constants $\alpha$ and $\tilde{C}$. It follows
 \[
  -\psi_t + \lambda \psi + \pare{\psi(y+1, t) - 2\psi(y,t) + \psi(y-1, t)} \beta > e^{-\alpha t} \pare{(\alpha + \lambda)(y^{2n} + \tilde{C}) - 2\beta y^{2n}}>0,
 \]
 when $\alpha + \lambda > 2 \beta$. Choosing $\alpha$ satisfying the previous inequality, then $v+\xi\psi$, for any $\xi>0$, is a viscosity supersolution to \eqref{eq: HJB lambda}.
 Once we show $u\leq v+ \xi \psi$, the statement of the lemma then follows after sending $\xi\downarrow 0$.

 Since both $u$ and $v$ have at most linear growth
 \begin{equation}\label{eq: limit u-v}
  \lim_{|y|\rightarrow \infty} (u-v-\xi \psi)(y,t) = -\infty.
 \end{equation}
 Replacing $v$ by $v+\xi \psi$, we can assume that $u$ (resp. $v$) is a viscosity subsolution (resp. supersolution) to \eqref{eq: HJB lambda} and
 \[
  \sup_{\Integer \times [0,1]}(u-v) = \sup_{\mathcal{O} \times [0,1]}(u-v), \quad \text{ for some compact set } \mathcal{O}\subset \Integer.
 \]
 Then $u\leq v$ follows from the standard argument in viscosity solutions (cf. e.g. \cite[Theorem 4.4.4]{Pham}), which we briefly recall below.

 Assume $M:= \sup_{\Integer \times [0,1]}(u-v) = \sup_{\mathcal{O} \times [0,1]}(u-v)>0$ and the maximum is attained at $(\overline{x}, \overline{t})\in \mathcal{O}\times [0,1]$. For any $\epsilon>0$, define
 \[
  \Phi_\epsilon(x,y,t,s) := u(x,t) - v(y,s) -\phi_\epsilon(x,y,t,s), \quad \text{ where } \phi_\epsilon(x,y,t,s) := \frac{1}{\epsilon}[|x-y|^2 + |t-s|^2].
 \]
 The upper-semicontinuous function $\Phi_\epsilon$ attains its maximum, denoted by $M_\epsilon$, at $(x_\epsilon, y_\epsilon, t_\epsilon, s_\epsilon)$. One can show, using the same argument as in \cite[Theorem 4.4.4]{Pham},
 \[
  M_\epsilon \rightarrow M \quad \text{ and } \quad  (x_\epsilon, y_\epsilon, t_\epsilon, s_\epsilon) \rightarrow (\overline{x}, \overline{x}, \overline{t}, \overline{t})\in \mathcal{O}^2\times[0,1]^2 \quad \text{as } \epsilon \downarrow 0.
 \]
 Here $(x_\epsilon, y_\epsilon, t_\epsilon, s_\epsilon) \in \mathcal{O}^2\times [0,1]^2$ for sufficiently small $\epsilon$. Now observe that
 \begin{itemize}
  \item $(x_\epsilon, t_\epsilon)$ is a local maximum of $(x,t) \mapsto u(x,t) - \phi_\epsilon (x, y_\epsilon, t, s_\epsilon)$;
  \item $(y_\epsilon, s_\epsilon)$ is a local minimum of $(y,t) \mapsto v(y,s) + \phi_\epsilon (x_\epsilon, y, t_\epsilon, s)$.
 \end{itemize}
 Then the viscosity subsolution property of $u$ and the supersolution property of $v$ imply, respectively,
 \begin{align*}
  &-\frac{2}{\epsilon}(t_\epsilon - s_\epsilon) + \lambda u(x_\epsilon, t_\epsilon) - \pare{u(x_\epsilon + 1, t_\epsilon) - 2u(x_\epsilon, t_\epsilon) + u(x_\epsilon, t_\epsilon)} \beta \leq 0,\\
  & -\frac{2}{\epsilon}(t_\epsilon - s_\epsilon) + \lambda v(y_\epsilon, s_\epsilon) - \pare{u(y_\epsilon + 1, s_\epsilon) - 2v(y_\epsilon, s_\epsilon) + v(y_\epsilon, s_\epsilon)} \beta \geq 0.
 \end{align*}
 Taking difference of the previous inequalities yields
 \[
  (\lambda + 2\beta) (u(x_\epsilon, t_\epsilon) - v(y_\epsilon, s_\epsilon)) \leq \beta \pare{u(x_\epsilon + 1, t_\epsilon) + u(x_\epsilon -1, t_\epsilon)} - \beta \pare{v(y_\epsilon + 1, s_\epsilon) + v(y_\epsilon -1, s_\epsilon)}.
 \]
 Sending $\epsilon \downarrow 0$ on both sides, we obtain
 \[
  (\lambda + 2\beta) M = (\lambda + 2 \beta) u(\overline{x}, \overline{t}) \leq \beta \pare{u(\overline{x}+1, \overline{t})- v(\overline{x}+1, \overline{t})} + \beta \pare{u(\overline{x}-1, \overline{t})- v(\overline{x}-1, \overline{t})} \leq 2\beta M,
 \]
 which contradicts with $\lambda M >0$.

\end{proof}

\bibliographystyle{siam}
\bibliography{biblio}

\end{document}